\documentclass[11pt]{article}

\usepackage[utf8]{inputenc}
\usepackage{amsmath,amssymb,amsthm}
\usepackage{geometry}
\usepackage{graphicx}
\usepackage{thmtools, thm-restate}
\usepackage[colorlinks,citecolor=blue,urlcolor=blue,linkcolor=blue,bookmarks=false]{hyperref}
\usepackage{listings}
\usepackage[framemethod=tikz]{mdframed}
\usepackage[sort]{natbib}
\usepackage{float}
\usepackage{makecell}
\usepackage{cancel}
\usepackage[shortlabels]{enumitem}
\usepackage{bm}
\usepackage{caption, subcaption}
\usepackage{booktabs}

\allowdisplaybreaks

\newtheorem{theorem}{Theorem}
\newtheorem{lemma}{Lemma}

\newtheorem{proposition}{Proposition}

\newtheorem{algorithm}{Algorithm}

\theoremstyle{definition}

\newtheorem{cumulative-example}{Cumulative trimming example}
\newtheorem{smooth-cumulative-example}{Cumulative smooth trimming example}

\theoremstyle{remark}
\newtheorem{assumption}{Assumption}
\newtheorem{remark}{Remark}

\def\inprob{\stackrel{p}{\rightarrow}}
\def\indist{\rightsquigarrow}
\newcommand\ind{\protect\mathpalette{\protect\independenT}{\perp}}
\def\independenT#1#2{\mathrel{\rlap{$#1#2$}\mkern4mu{#1#2}}}

\DeclareSymbolFont{bbold}{U}{bbold}{m}{n}
\DeclareSymbolFontAlphabet{\mathbbold}{bbold}
\newcommand{\one}{\mathbbold{1}}

\def\bbE{\mathbb{E}}
\def\bbP{\mathbb{P}}
\def\bbR{\mathbb{R}}
\def\bbV{\mathbb{V}}

\usepackage{stackengine}  
\stackMath                

\title{{\fontsize{15pt}{18pt}\selectfont\textbf{Propensity score weighting across counterfactual worlds: longitudinal effects under positivity violations}}}

\author{\Large Alec McClean\footnote{Corresponding author: \texttt{hadera01@nyu.edu}}\, and Iv\'{a}n D\'{i}az \\[0.5em]
\emph{Division of Biostatistics, New York University Grossman School of Medicine}}

\begin{document}
\maketitle

\begin{abstract}
    When examining a contrast between two interventions, longitudinal causal inference studies frequently encounter positivity violations when one or both regimes are impossible to observe for some subjects. Existing weighting methods either assume positivity holds or produce effects that conflate interventions' impacts on ultimate outcomes with their effects on intermediate treatments and covariates. We propose a novel class of estimands---cumulative cross-world weighted effects---that weights potential outcome differences using propensity scores adapting to positivity violations cumulatively across timepoints and simultaneously across both counterfactual treatment histories. This new estimand isolates mechanistic differences between treatment regimes, is identifiable without positivity assumptions, and circumvents the limitations of existing longitudinal methods. Further, our analysis reveals two fundamental insights about longitudinal causal inference under positivity violations. First, while mechanistically meaningful, these effects correspond to non-implementable interventions, exposing a core interpretability-implementability tradeoff. Second, the identified effects faithfully capture mechanistic differences only under a partial common support assumption; violations cause the identified functional to collapse to zero, even when the causal effect is non-zero. We develop doubly robust-style estimators that achieve asymptotic normality and parametric convergence under nonparametric assumptions on the nuisance estimators. To this end, we reformulate challenging density ratio estimation as regression function estimation, which is achievable with standard machine learning methods. We illustrate our methods through analysis of union membership's effect on earnings.
\end{abstract}

\section{Introduction}

There is a large and fast-growing literature in causal inference for estimating treatment effects under violations of the positivity assumption. In longitudinal data, positivity requires that every subject has a non-zero probability of receiving every treatment regime under consideration. In longitudinal data, this can become increasingly difficult to satisfy with more timepoints as the number of treatment regimes increases exponentially \citep{hernan2020whatif}. While positivity violations are challenging in single timepoint data, longitudinal data present unique difficulties that existing methods struggle to address.

\bigskip

A popular approach to addressing positivity violations is to weight or trim the dataset to down-weight or exclude subjects with small probabilities of receiving the treatment regimes under which one would like to estimate mean potential outcomes \citep{crump2009dealing}. While well-studied with single timepoint data, standard methods do not easily generalize to address non-baseline positivity violations in longitudinal data and remain vulnerable to such violations. The core challenge is that violations occurring after baseline (non-baseline positivity violations) cannot be addressed through weighting on baseline covariates, but direct conditioning on time-varying covariates risks inducing collider bias by conditioning on post-treatment variables \citep{jensen2024identification, petersen2012diagnosing}. Existing longitudinal weighting methods typically sidestep this challenge by assuming positivity holds, and focus on weighting for minimizing estimation error rather than addressing violations of positivity \citep{imai2015robust, zeng2023propensity, viviano2021dynamic}.  An important exception is \citet{schomaker2024causal}, which considered weighting towards a particular longitudinal dose-response curve while addressing positivity violations, but this approach did not consider contrasts in effects, which are usually the target of interest with weighting.

\bigskip

More successful approaches for addressing positivity violations in longitudinal data have stemmed from ``stochastic longitudinal modified treatment policies" (S-LMTPs). First introduced in \citet{robins2004effects} and \citet{stock1989nonparametric}, S-LMTPs have subsequently evolved through various contributions \citep{van2007causal, haneuse2013estimation, diaz2012population, kennedy2019nonparametric, diaz2023nonparametric, richardson2013single, young2014identification}. Among these, a recent development for weighting is the ``flip" intervention, which flips subjects toward a target treatment regime using weights constructed from non-baseline covariates \citep{mcclean2025longitudinal}. This construction allows for dynamic weighting and trimming on non-baseline covariates, making it possible to avoid positivity violations at each timepoint.  However, flip interventions face an inherent limitation: because they are ``single-world'' interventions (in the sense of \citet{richardson2013single}) that sequentially adapt to prior counterfactual interventions, contrasts between potential outcomes under two sequences of flip interventions do not isolate the difference in mean potential outcomes under the two static regimes of interest. The dynamic nature of flip interventions means that the intervention at timepoint $t$ affects not only the final outcome but also subsequent covariates and treatments, which in turn affect future flip interventions. Consequently, differences in flip effects can reflect the intervention's impact on the entire treatment process rather than just the ultimate outcome, with the corresponding implications for interpretation.

\subsection{Our contributions}

To overcome the limitations of standard weighting methods in longitudinal settings and the drawbacks of contrasts of flip effects, we introduce the \emph{cumulative cross-world weighted effect}:

\begin{equation} \label{eq:cumulative-weight}
    \psi(\overline a_T, \overline a_T^\prime) := \bbE \left( \left\{ Y(\overline a_T) - Y(\overline a_T^\prime) \right\} \prod_{t=1}^T w_t \left[ p_t\{ \overline X_t(\overline a_{t-1}) \} \right] w_t^\prime \left[ p_t^\prime \{ \overline X_t(\overline a_{t-1}^\prime) \} \right] \right),
\end{equation} 
where $p_t\{ \overline X_t(\overline a_{t-1}) \} = \bbP \{ A_t(\overline a_{t-1}) = a_t \mid \overline X_{t-1}(\overline a_{t-1}) \}$ and $p_t^\prime \{ \overline X_t(\overline a_{t-1}^\prime) \} = \bbP \{ A_t(\overline a_{t-1}^\prime) = a_t^\prime \mid X_t(\overline a_{t-1}^\prime) \}$ are the natural propensity scores under interventions $\overline a_{t-1}$ and $\overline a_{t-1}^\prime$, respectively, and $w_t, w_t^\prime$ are weighting functions.

\bigskip

This effect re-weights each unit using the product of \emph{natural} propensity scores that \emph{would occur} under \emph{both} treatment regimes of interest.  The resulting unit-level weight adapts to non-baseline positivity violations and isolates the mean difference of the two potential outcomes of interest. To our knowledge, this is a novel causal estimand that corresponds to the effect researchers often have in mind when conducting longitudinal weighting or trimming.

\bigskip

We refer to $\psi(\overline a_T,\overline a_T')$ as a \emph{cumulative cross-world} weighted effect.  
The term \emph{cumulative} reflects the product over time indices $t=1,\dots,T$, while \emph{cross-world} denotes the combination of propensity scores from the two counterfactual histories generated by interventions $\overline a_{t-1}$ and $\overline a_{t-1}^\prime$.  Through the choice of weight functions $w_t$ and $w_t^\prime$, this quantity can capture the desired mean difference in potential outcomes while adapting to non‐baseline positivity violations.

\bigskip

We examine the cumulative cross-world weighted effect in \eqref{eq:cumulative-weight} in detail, first giving several example weights and then discussing its benefits and drawbacks compared to other approaches. We show that while this weighted effect is ``mechanism-relevant''---isolating the mean difference in potential outcomes under the two regimes of interest while adapting to positivity violations---it is not ``policy-relevant'' since it represents a ``cross-world'' effect that cannot be implemented as a practical intervention.

\bigskip

We establish identification of the cumulative cross-world weighted effect, showing that this is possible without the positivity assumption, but requires the strong sequential randomization assumption, a stronger exchangeability assumption than necessary for identifying many longitudinal effects.  We then develop estimation methods. For smooth weights, such as overlap weights or smooth trimming weights, we develop efficient doubly robust-style estimators based on the efficient influence function of the weighted effect. Consequently, we establish that weak convergence is feasible under nonparametric assumptions on the relevant nuisance functions. These include a conditional density ratio across time-varying covariates, which is typically difficult to estimate. We show how this ratio can be reformulated as a ratio of regression functions, which can be estimated with typical machine learning methods. 

\bigskip

Finally, a crucial conceptual clarification we provide is a careful examination of conditions under which cumulative cross-world weighted effects remain meaningful. We explicitly formalize a partial common support assumption, distinct from standard positivity, ensuring covariate distributions overlap sufficiently across regimes. We show this assumption is necessary to ensure the identified functional faithfully represents the underlying causal contrast, clearly delineating when weighted effects accurately capture the mechanistic relationships of interest.

\subsection{Structure of the paper}

Section~\ref{sec:notation} provides generic mathematical notation while Section~\ref{sec:longitudinal-setup} introduces our notation and setup for longitudinal data, including causal assumptions. Section~\ref{sec:weighted-effect} re-introduces the cumulative cross-world weighted effect, provides several example weight functions, and discusses its benefits and drawbacks.  Section~\ref{sec:identification} identifies the effect while Section~\ref{sec:estimation} proposes a doubly robust-style estimator.  Section~\ref{sec:data-analysis} illustrates our methods with a data analysis and Section~\ref{sec:discussion} concludes and discusses avenues for future work.

\subsection{Mathematical notation} \label{sec:notation}

For a function $f(Z)$, we use $\lVert f \rVert = \sqrt{\int f(z)^2 d\bbP(z)}$ to denote the $L_2(\bbP)$ norm, $\lVert f \rVert_p = \left( \int f(z)^p d\bbP(z) \right)^{1/p}$ to denote the generic $L_p(\bbP)$ norm for $p \geq 1$, $\bbP(f) = \int_{\mathcal{Z}} f(z) d\bbP(z)$ to denote the average with respect to the underlying distribution $\bbP$, and $\bbP_n(f) = \frac{1}{n} \sum_{i=1}^{n} f(Z_i)$ to denote the empirical average with respect to $n$ observations.  In a standard abuse of notation, when $B$ is an event we let $\bbP(B)$ denote the probability of $B$.  We also denote expectation and variance with respect to the underlying distribution by $\bbE$ and $\bbV$, respectively.  We use $a \wedge b$ for minimum and $a \vee b$ for maximum, and $a \lesssim b$ to mean $a \leq Cb$ for some constant $C$. We use $\indist$ to denote convergence in distribution, and $\inprob$ for convergence in probability. Additionally, we use $o_\bbP(\cdot)$ to denote convergence in probability to zero, i.e., if $X_n$ is a sequence of random variables then $X_n = o_\bbP (r_n)$ implies $\left| \frac{X_n}{r_n} \right| \inprob 0$.

\section{Setup for longitudinal data} \label{sec:longitudinal-setup}

We assume $n$ observations drawn iid from some distribution $\bbP$ in a space of distributions $\mathcal{P}$; i.e., we observe data $\{ Z_i \}_{i=1}^{n} \stackrel{iid}{\sim} \bbP \in \mathcal{P}$. We assume each observation consists of longitudinal data over $T$ timepoints, so that
\[
Z = (X_1, A_1, X_2, A_2, \dots, X_T, A_T, Y),
\]
where $X_t \in \bbR^d$ are time-varying covariates ($X_1$ are baseline covariates), $A_t \in \{0,1\}$ is a time-varying binary treatment, and $Y \in \bbR$ is the ultimate outcome of interest. For a time-varying random variable $O_t$, let $\overline O_t = (O_1, \dots, O_t)$ denote its history up to time $t$ and $\underline O_t = (O_t, \dots, O_T)$ denote its future from time $t$. Let $H_t = (\overline X_t, \overline A_{t-1})$ denote covariate and treatment history up until treatment in timepoint $t$. 

\bigskip

We formalize the definition of causal effects using a nonparametric structural equation model (NPSEM) \citep{pearl2009causality}. We assume the existence of deterministic functions $\{ f_{X,t}, f_{A,t}\}_{t=1}^{T}$ and $f_Y$ such that 
\begin{align*}
    X_t &= f_{X,t} ( A_{t-1}, H_{t-1}, U_{X,t} ), \\
    A_t &= f_{A,t} ( H_t, U_{A,t} ) \text{, and } \\
    Y &= f_Y ( A_T, H_T, U_Y ).
\end{align*}
Here, $\Big\{ \big\{ U_{X,t}, U_{A,t}: t \in \{1, \ldots, T\} \big\}, U_Y \Big\}$ is a vector of exogenous variables. Subsequently, we'll define restrictions on their joint distribution that facilitate identification of causal effects. We will define the effects in terms of hypothetical interventions in which equation $A_t = f_{A,t}(H_t, U_{A,t})$ is removed from the structural model and the exposure is assigned as a new treatment $a_t$. An intervention that sets exposures up to time $t-1$ to $\overline a_{t-1}$ generates counterfactual covariates $X_t(\overline a_{t-1}) = f_{X,t} \big\{ a_{t-1}, H_{t-1}(\overline a_{t-2}), U_{X,t} \big\}$, where the counterfactual treatment and covariate history is defined recursively as $H_t(\overline a_{t-1}) = \{ X_t(\overline a_{t-1}), a_{t-1}, H_{t-1}(\overline a_{t-2}) \}$.  The variable $A_t(\overline a_{t-1})$ is called the \emph{natural} treatment value \citep{richardson2013single,   young2014identification}, and represents the possibly counterfactual value of the treatment that would have been observed at time $t$ under an intervention carried out up to time $t-1$ but discontinued thereafter. Similarly, we refer to  $X_t(\overline a_{t-1})$ as the \textit{natural covariate} value. Moreover, we define the counterfactual covariate history at time $t$ under intervention $\overline a_{t-1}$ as $\overline X_t(\overline a_{t-1}) = \big\{ X_s(\overline a_{s-1}) \big\}_{s=1}^{t}$. Finally, an intervention in which all treatment variables up to $t=T$ are intervened on generates a counterfactual outcome $Y \left( \overline a_T \right) = f_Y \big\{ a_T, H_T(\overline a_{T-1}), U_Y \big\}$. 

\subsection{Causal assumptions}

The NPSEM implicitly encodes the consistency assumption due to the modular definition of counterfactuals within the framework, and because one subject's information does not depend on another's. This assumption would be violated if there were interference between subjects \citep{tchetgen2012causal}. We require the following exchangeability assumption on the exogeneous variables.
\begin{assumption}[Strong sequential randomization] \label{asmp:strong-seq-exch}
$U_{A,t} \ind\! \{ \underline U_{X,t+1},\! \underline U_{A,t+1},\! U_Y \!\} \mid\! H_t \ \forall \ t \leq T$.
\end{assumption}
Assumption~\ref{asmp:strong-seq-exch} is satisfied if common causes of treatment $A_t$ and future covariates \emph{and treatments} are measured. This assumption is similar to that required by \citet{richardson2013single} (cf. Theorem 31).  Finally, note that we do not require the typical positivity assumption, which says that time-varying propensity scores are bounded away from zero and one ($0 < \bbP(A_t = 1 \mid H_t) < 1$), because we will consider longitudinal weights which adapt to positivity violations.

\section{Cumulative cross-world weighted treatment effects} \label{sec:weighted-effect}

Standard weighting methods fail to address non-baseline positivity violations in longitudinal data, while existing alternatives like flip interventions introduce additional effects on future covariates and treatments beyond isolating the causal effect of interest. Here, we examine our proposed weighted treatment effect in more detail. Recall that it is defined as 

\begin{equation} \tag{\ref{eq:cumulative-weight}}
    \bbE \left( \left\{ Y(\overline a_T) - Y(\overline a_T^\prime) \right\} \prod_{t=1}^T w_t \left[ p_t\{ \overline X_t(\overline a_{t-1}) \} \right] w_t^\prime \left[ p_t^\prime \{ \overline X_t(\overline a_{t-1}^\prime) \} \right] \right).
\end{equation}
\normalsize
The key components of this estimand are:
\begin{itemize}
    \item $\overline a_T$ and $\overline a_T^\prime$ are the two binary treatment regimes of interest over $T$ timepoints;
    \item $Y(\overline a_T) - Y(\overline a_T^\prime)$ is the difference in potential outcomes under the two regimes, which isolates the individual-level causal effect of interest;
    \item $w_t(\cdot)$ and $w_t^\prime(\cdot)$ are time-specific weight functions that can be chosen to address positivity violations;
    \item $p_t\{ \overline X_t(\overline a_{t-1}) \} = \bbP\left\{ A_t(\overline a_{t-1}) = a_t \mid \overline X_t (\overline a_{t-1}) \right\}$ and $p_t^\prime \{ \overline X_t(\overline a_{t-1}^\prime) \} = \bbP \left\{ A_t(\overline a_{t-1}^\prime) = a_t^\prime \mid \overline X_t(\overline a_{t-1}^\prime) \right\}$ are the natural propensity scores under regimes $\overline a_{t-1}$ and $\overline a_{t-1}^\prime$, respectively; and
    \item $\overline X_t(\overline a_{t-1})$ and $\overline X_t(\overline a_{t-1}^\prime)$ are the natural covariate histories that would have evolved under each regime up to time $t$.
\end{itemize}
The fundamental insight underlying this weighted effect is that it achieves two distinct goals simultaneously: it isolates the causal contrast of interest, $Y(\overline a_T) - Y(\overline a_T^\prime)$, and addresses positivity violations through the cumulative weights.

\bigskip

There are a variety of weight functions that one might consider. While one could attempt to balance directly on counterfactual covariates---as in \citet{viviano2021dynamic}, which requires additional local structural model assumptions---we focus on weights that depend on the propensity score. A simple example is trimming weights:
\begin{align*}
    w_t \left[ p_t\{ \overline X_t(\overline a_{t-1}) \} \right] &= \one \big[ p_t\{ \overline X_t(\overline a_{t-1}) \} > \varepsilon \big], \\
    w_t^\prime \left[ p_t^\prime \{ \overline X_t(\overline a_{t-1}^\prime) \} \right] &= \one \big[ p_t^\prime \{ \overline X_t(\overline a_{t-1}^\prime) \} > \varepsilon^\prime \big],
\end{align*}
which trim to subjects whose natural propensity scores would both be greater than some thresholds $\varepsilon, \varepsilon^\prime \geq 0$. Table~\ref{tab:weights} gives several other examples.

\begin{table}[ht] 
    \centering 
    \begin{tabular}{p{5cm} p{4.5cm} p{4.5cm}}
        \toprule
        \textbf{Type of weighting} & $w_t \left( p_t \right)$ & $w_t^\prime \left( p_t^\prime \right)$ \\
        \midrule
        No weighting & \( 1 \) & \( 1 \)  \\
        \addlinespace
        Weighting towards $\overline a_T$ only & \( p_t \) & 1 \\
        \addlinespace
        Weighting towards $\overline a_T^\prime$ only & \( 1 \) & \( p_t^\prime \) \\
        \addlinespace
        Overlap weighting & \( p_t \) & \( p_t^\prime \) \\
        \addlinespace
        Trimming & \( \one ( p_t \geq \varepsilon ) \) for $\varepsilon \geq 0$ & \( \one ( p_t^\prime \geq \varepsilon^\prime ) \) for $\varepsilon^\prime \geq 0$ \\
        \addlinespace
        Smooth trimming & \( w_t \big( p_t  ; \varepsilon \big) \), where $w_t(x; \varepsilon)$ approximates $\one (x \geq \varepsilon)$ & \( w_t \big( p_t^\prime ; \varepsilon \big) \), where $w_t(x; \varepsilon)$ approximates $\one (x \geq \varepsilon)$\\
        \bottomrule
    \end{tabular}
    \caption{Weights for longitudinal weighting and trimming}
    \label{tab:weights}
\end{table}

\subsection{Benefits and drawbacks}

The weighted effect in \eqref{eq:cumulative-weight} offers distinct advantages and disadvantages compared to alternatives like flip interventions, illustrating an inherent tradeoff between mechanism-relevance and policy-relevance.

\bigskip

Here, we assert a ``mechanism-relevant'' effect would be the effect that best reflects the causal contrast $Y(\overline a_T) - Y(\overline a_T^\prime)$. The design of the weighted effect in \eqref{eq:cumulative-weight} maximizes its ``mechanism-relevance'' by isolating the causal contrast $Y(\overline a_T) - Y(\overline a_T^\prime)$ while addressing positivity violations through the weight functions.  This focus on the potential outcome difference is illustrated by the following null preservation property:

\begin{proposition}
    If $\bbP \{ Y(\overline a_T) = Y(\overline a_T^\prime) \} = 1$ then $\psi(\overline a_T, \overline a_T^\prime) = 0$.
\end{proposition}
This result shows that the cumulative cross-world weighted effect is zero whenever the two potential outcomes under consideration are almost surely equal. This is a useful result. In particular, alternative proposals, like flip effects, do not satisfy it (cf. , \citet{mcclean2025longitudinal}, Section 4.2). As discussed in \citet{mcclean2025longitudinal}, because flip interventions affect subsequent treatments and outcomes as well as the ultimate outcome of interest, the difference in mean potential outcomes under two flip interventions will typically fail to isolate the difference in the two potential outcomes of interest. 

\bigskip

A second related advantage is that the cumulative cross-world weighted effect corresponds to a weighted population of interest. Each subject has their own weight value
\[
\prod_{t=1}^{T} w_t \left[ \bbP \left\{ A_{t,i}(\overline a_{t-1}) = a_t \mid \overline X_{t,i}(\overline a_{t-1}) \right\} \right] w_t^\prime \left[ \bbP \left\{ A_{t,i}(\overline a_{t-1}^\prime) = a_t^\prime \mid \overline X_{t,i}(\overline a_{t-1}^\prime) \right\} \right],
\]
where $\overline X_{t,i}(\overline a_{t-1}), \overline X_{t,i} (\overline a_{t-1}^\prime)$ are the natural covariate and treatment history for unit $i$. Supposing one can identify these weights from the observed data (see Section~\ref{sec:identification}), it is possible to assign a weight to each individual in the study and examine summary statistics of the weights.

\bigskip

However, this mechanism-relevance comes at a cost in terms of ``policy-relevance''---whether a causal effect corresponds to an intervention that could be implemented in practice.  The weighted treatment effect in \eqref{eq:cumulative-weight} is ``cross-world'' because it incorporates a product of weights under two different regimes, $\overline a_t$ and $\overline a_t^\prime$. As a result, the interventions that yields these effects rely on cross-world propensity scores, which one could never observed in practice; e.g., for one individual, we cannot construct an experiment to yield both $X_t(\overline a_{t-1})$ and $X_t(\overline a_{t-1}^\prime)$ simultaneously. Consequently, these effects cannot be falsified experimentally. This limitations parallels natural effects in mediation \citep{andrews2021insights, richardson2013single}.  

\bigskip

A related disadvantage is that a stronger assumption is required for identification. Specifically, we require strong sequential randomization, in Assumption~\ref{asmp:strong-seq-exch}. By contrast, flip interventions can be constructed so they are identifiable under only standard sequential exchangeability (see, \citet{mcclean2025longitudinal}, Theorem 1 and subsequent remarks).

\begin{remark}
    In this section, we have emphasized the tradeoff between mechanism-relevance and policy-relevance. This tradeoff arises only because of positivity violations and because we allow interventions to affect subsequent covariates as well as the outcome of interest. If positivity were satisfied for both $\overline a_T$ and $\overline a_T^\prime$, then one could identify and estimate $\bbE \{ Y(\overline a_T) - Y(\overline a_T^\prime) \}$, which is both policy-relevant and mechanism-relevant. Or, if interventions only affected outcomes, but not intermediate covariates, then one could know that conditional covariate densities are constant across different treatment conditioning sets and therefore the cross-world nature of $\overline X_t(\overline a_{t-1})$ would disappear, because $\overline X_t(\overline a_{t-1})$ and $\overline X_t(\overline a_{t-1}^\prime)$ would be equal in distribution, by assumption. Since neither simplification typically holds, this tradeoff arises frequently in longitudinal data.
\end{remark}

\begin{remark}
    We hesitate to recommend whether the cumulative cross-world weighted effect is categorically better or worse than alternatives like flip effects. Rather, the key insight is that different approaches involve fundamental tradeoffs, and the choice depends on whether mechanism-relevance or policy-relevance is prioritized for a given research question.
\end{remark}

\section{Identification} \label{sec:identification}

This section establishes identification of the cumulative cross-world weighted effect $\psi(\overline a_T, \overline a_T^\prime)$, from \eqref{eq:cumulative-weight}. Our identification strategy proceeds in two steps, reflecting the cross-world nature of the estimand: first, we identify the natural propensity scores that appear in the weight functions; second, we establish identification of the cumulative cross-world effect with a weighted g-formula.

\bigskip

The weight functions depend on natural propensity scores of the form $\bbP \{ A_t(\overline a_{t-1}) = a_t \mid \overline X_t(\overline a_{t-1}) \}$ and $\bbP \{ A_t(\overline a_{t-1}^\prime) = a_t^\prime \mid \overline X_t(\overline a_{t-1}^\prime)  \}$. These require separate identification within the main identification result. To streamline the presentation, we introduce simplified notation for the identified analogs to these natural propensity scores:
\begin{align}
    \pi_t(\overline X_t) &= \bbP (A_t = a_t \mid \overline X_t, \overline A_{t-1} = \overline a_{t-1}) \text{ and } \\
    \pi_t^\prime(\overline X_t) &= \bbP (A_t = a_t^\prime \mid \overline X_t, \overline A_{t-1} = \overline a_{t-1}^\prime).
\end{align}
Moreover, for the rest of the paper, by convention we set the propensity scores to zero whenever the conditioning event has probability zero. The next result establishes conditions for identification.
\begin{lemma} \label{lem:id-prop-scores}
    Under the NPSEM and Assumption~\ref{asmp:strong-seq-exch}. Then, 
    \begin{enumerate}
        \item If $\pi_s(\overline X_s) > 0$ for all $s < t$ then $\bbP \{ A_t(\overline a_{t-1}) = a_t \mid \overline X_t(\overline a_{t-1})  \} = \pi_t(\overline X_t)$, and 
        \item if $\pi_s^\prime(\overline X_s) > 0$ for all $s < t$ then $\bbP \{ A_t(\overline a_{t-1}^\prime) = a_t^\prime \mid \overline X_t(\overline a_{t-1}^\prime)  \} = \pi_t^\prime(\overline X_t)$.
    \end{enumerate}
\end{lemma}
All proofs are delayed to the appendix. This result shows that, supposing positivity holds for all prior timepoints, then the natural propensity score at timepoint $t$ is identified.  We elaborate further on this result  after stating the main result, Theorem~\ref{thm:id-cumulative}, which identifies the full weighted effect.

\begin{theorem}[Identification of cumulative cross-world trimmed effects] \label{thm:id-cumulative}
    Suppose the NPSEM and Assumption~\ref{asmp:strong-seq-exch} hold. Moreover, suppose the weights satisfy two conditions:
    \begin{enumerate}[left=0pt]
        \item  $\pi_t(\overline x_t) \pi_t^\prime(\overline x_t) = 0 \implies w_t\{ \pi_t(\overline x_t) \} w_t^\prime\{ \pi_t^\prime(\overline x_t) \} = 0$. \label{cond:positivity}
    \end{enumerate}
    Then,
    \begin{align}
        \psi(\overline a_T, \overline a_T^\prime) &= \int_{\overline{\mathcal{X}}_T} \bbE ( Y \mid \overline a_T, \overline x_T) \prod_{t=1}^T w_t\{ \pi_t(\overline x_t) \}w_t^\prime\{ \pi_t^\prime(\overline x_t)\} d\bbP(x_t \mid \overline A_{t-1} = \overline a_{t-1}, \overline x_{t-1}) \label{eq:id-1} \\
        &- \int_{\overline{\mathcal{X}}_T} \bbE ( Y \mid \overline a_T^\prime, \overline x_T) \prod_{t=1}^T w_t\{ \pi_t(\overline x_t) \}w_t^\prime\{ \pi_t^\prime(\overline x_t)\} d\bbP(x_t \mid \overline A_{t-1} = \overline a_{t-1}^\prime, \overline x_{t-1}). \label{eq:id-2}
    \end{align}    
\end{theorem}

This result establishes when the cumulative cross-world weighted effect is identifiable as a difference in g-formula functionals \citep{robins1986new}. To our knowledge, this is a new type of identification result: it blends the traditional g-formula with cross-world weights $w_t \{ \pi_t(\overline x_t) \} w_t^\prime\{ \pi_t^\prime(\overline x_t) \}$, which behave like the intervention propensity scores that arise with dynamic stochastic interventions or S-LMTPs (eg., \citet{kennedy2019nonparametric}, Lemma 1). However, these weights don't define an intervention, but rather define a weighted population that facilitates cross-world comparisons.

\bigskip

We now discuss the assumptions required for identification. In addition to the standard NPSEM assumption, we also impose strong sequential randomization (Assumption~\ref{asmp:strong-seq-exch}) to identify cross-world propensity scores in Lemma~\ref{lem:id-prop-scores}. This is a stronger assumption than typically necessary for identification. It allows us to relate the counterfactual and cross-world conditioning sets for one of the propensity scores to the observed conditioning sets.

\bigskip

Notably, Theorem~\ref{thm:id-cumulative} does not require a positivity assumption.  Instead, positivity is guaranteed by condition~\ref{cond:positivity}.  It asserts that the weight product $w_t(\pi_t) w_t^\prime(\pi_t^\prime)$ is zero whenever either propensity score used to construct the weights is zero. This is enforceable by the design of the weights. From Table~\ref{tab:weights}, the overlap weights and trimming weights satisfy this directly, and the smooth trimming weights can satisfy this constraint with specific choices of smoothing function. 

\bigskip

Condition~\ref{cond:positivity} is crucial because it facilitates identification of the natural propensity scores within the g-formula result in Theorem~\ref{thm:id-cumulative}. Informally, presuming one has successfully identified everything before timepoint $t$, the positivity assumptions in Lemma~\ref{lem:id-prop-scores} hold within the innermost expectation because otherwise the weights from prior timepoints ensure the overall g-formula is zero. Therefore, one can identify the natural propensity scores at timepoint $t$ using Lemma~\ref{lem:id-prop-scores}. One can apply this recursively to identify the full g-formula, as in Theorem~\ref{thm:id-cumulative}, but without requiring the positivity assumptions in Lemma~\ref{lem:id-prop-scores}.   

\subsection{Common support of time-varying covariates}

A critical consideration concerns the overlap of covariate distributions across treatment regimes. When the conditional laws $X_t \mid \overline a_{t-1},\overline x_{t-1}$ and $X_t \mid \overline a_{t-1}^{\prime},\overline x_{t-1}$ have disjoint support, the identified functionals in \eqref{eq:id-1} and \eqref{eq:id-2} vanish, yielding $\psi(\overline a_T,\overline  a_T^{\prime})=0$ even when the true potential outcome contrast is non-zero. This occurs because, by convention, the propensity scores are set to zero (and consequently their weights to zero) whenever the conditioning event has probability zero, effectively trimming away paths with no common support. Therefore, an overlap condition is required not for identification itself, but rather for $\psi(\overline a_T, \overline a_T^\prime)$ to be an informative functional about the mechanistic difference in potential outcomes. To formalize this requirement, we first define the following covariate density ratio:
\[
\rho_t(\overline X_t) = \frac{d\bbP(X_t \mid \overline A_{t-1} = \overline a_{t-1}, \overline X_{t-1})}{d\bbP(X_t \mid \overline A_{t-1} = \overline a_{t-1}^\prime, \overline X_{t-1})}.
\]
The following result establishes when the weighted functional $\psi(\overline a_T, \overline a_T^\prime)$ faithfully represents differences in potential outcomes.
\begin{proposition}[Informativeness of weighted effects] \label{prop:informativeness}
    Suppose the NPSEM and Assumption~\ref{asmp:strong-seq-exch} hold, and condition~\ref{cond:positivity} is satisfied. Then:
    \begin{enumerate}
        \item If $\bbP \{ \rho_t(\overline X_t) \in (0, \infty) \} = 0$ for some $t$, then $\psi(\overline a_T, \overline a_T^\prime) = 0$ regardless of the true potential outcome contrast.
        \item If $\bbP \{\rho_t(\overline X_t) \in (0, \infty) \} > 0$ for all $t$ and $\bbP \{ \pi_s(\overline X_s) > 0, \pi_s^\prime(\overline X_s) > 0 \text{ for all } s \leq t \mid \rho_s(\overline X_s) \in (0, \infty) \text{ for all } s \leq t \} > 0$ for all $t$, then 
        \begin{align*}
            &\bbP \{ Y(\overline a_T) > Y(\overline a_T^\prime) \} = 1 \implies \psi(\overline a_T, \overline a_T^\prime) > 0 \text{ and}  \\
            &\bbP \{ Y(\overline a_T) < Y(\overline a_T^\prime) \} =1  \implies \psi(\overline a_T, \overline a_T^\prime) < 0.
        \end{align*}
    \end{enumerate}
\end{proposition}

This result shows that when covariate supports are completely disjoint, the weighted functional is uninformative because all paths receive zero weight in the identified functional. However, even partial overlap enables the functional to detect true differences. This result is reminiscent of the notion of comparability in the multi-valued treatment literature \citep{mcclean2024fair}. To ensure $\psi(\overline a_T, \overline a_T^\prime)$ is informative about mechanistic differences, we assume:
\begin{assumption}[Partial common support] \label{asmp:partial-common-support}
    For every $t$, \( \bbP\bigl\{ \rho_t(\overline X_t) \in (0, \infty)\bigr\} > 0. \)
\end{assumption}
Assumption~\ref{asmp:partial-common-support} requires that the covariate distributions have common support on a set of positive probability. This is considerably weaker than requiring full common support almost surely---it allows for subjects who would experience one treatment regime but not the other, as long as some subjects have overlapping covariate distributions under both regimes. Crucially, this assumption can be examined empirically by examining estimates for $\rho_t$. We demonstrate this in our data analysis in Section~\ref{sec:data-analysis}.
\bigskip

It is important to emphasize that these covariate overlap issues are conceptually distinct from positivity violations. To illustrate this, consider two simple data-generating processes, where we want to calculate the weighted difference in means of $Y(1,1)$ to $Y(0,0)$:
\begin{enumerate}
    \item Suppose $X_1 \sim \text{Unif}(0,1)$, $A_1 \sim \text{Bern}(0.5)$, $X_2 = A_1$, and $A_2 \sim \text{Bern}(0.5)$. This process exhibits no positivity violations, yet the conditional distributions $X_2 \mid A_1 = 1, X_1$ and $X_2 \mid A_1 = 0, X_1$ have disjoint support.
    \item Conversely, suppose $X_1 \sim \text{Unif}(0,1), A_1 \mid X_1 \sim \text{Bern} \bigl\{ \one(X_1>0.2)\bigr\}, X_2 \sim \text{Unif}(0,1),\, A_2 \mid X_2, A_1, X_1 \sim \text{Bern} \bigl\{ \one(X_1 X_2 - A_1 >0.2 ) \bigr\}$. Here $X_2\mid A_1=1$ and $X_2\mid A_1=0$ are both $\text{Unif}(0,1)$, so $\rho_2 \equiv 1$ (full covariate overlap). However, $A_1=1$ occurs only when $X_1>0.2$, so $\pi_1(X_1) \equiv \bbP(A_1 = 1 \mid X_1) =0$ on the set $\{ X_1 \leq 0.2 \}$.  Likewise $\pi_2(X_1,X_2) \equiv \bbP(A_2 = 1 \mid \overline X_2, A_1 = 1) = 0$ on the event $\{ X_1X_2 \leq 1.2 \}$. Similarly, $\pi_1^\prime(X_1) \equiv \bbP(A_1 = 0 \mid X_1) = 0$ on the set $\{ X_1 > 0.2\}$ and $\pi_2^\prime (X_1, X_2) \equiv \bbP(A_2 = 0 \mid \overline X_2, A_1 = 0) = 0$ on the set $\{ X_1 X_2 > 0.2\}$.   Thus the positivity assumptions fail on regions of positivity probability even though covariate overlap is perfect.
\end{enumerate}
These examples demonstrate that covariate overlap concerns arise independently of treatment assignment mechanisms.

\section{Estimation} \label{sec:estimation}

In this section, we outline a doubly robust-style estimator for the identified cumulative cross-world weighted effect. For simplicity, we focus on the first component of equation~\eqref{eq:id-1} in Theorem~\ref{thm:id-cumulative}, as our results immediately extend to the full contrast by linearity. We also focus on weight functions that are known smooth functions of the underlying propensity scores; that is, $w_t$ and $w_t^\prime$ are twice differentiable with non-zero and bounded derivatives. When the weight functions are unknown---for example, if one wanted to decide a trimming threshold or smooth trimming parameter data-adaptively---estimation and inference become more complex \citep{khan2022doubly}. We focus on smooth weight functions because this allows for the derivation of $\sqrt{n}$-consistent and asymptotically normal estimators under nonparametric assumptions by leveraging nonparametric efficiency theory and efficient influence functions \citep{bickel1993efficient}. Overlap weights and smooth trimming weights are appropriately smooth. By contrast, when the weights are non-smooth, such as with trimming weights, the weighted effect is not pathwise differentiable and the performance of estimators is dictated by the behavior of propensity score estimators within the trimming indicator. However, if trimming weights are desired, two options are available. First, they can be replaced by smooth approximations. Second, one can target the data-dependent estimand with estimated weights rather than with true weights \citep{van2007causal}.

\subsection{Notation}

To facilitate exposition, we refine our notation. First, let $\psi(\overline a_T)$ denote the first half of the identified cumulative cross-world weighted effect; i.e., 
\begin{equation} \label{eq:first-half}
    \psi(\overline a_T) = \int_{\overline{\mathcal{X}}_T} \bbE ( Y \mid \overline a_T, \overline x_T) \prod_{t=1}^T w_t \{\pi_t(\overline x_t) \} w_t^\prime\{ \pi_t^\prime(\overline x_t) \} d\bbP(x_t \mid \overline A_{t-1} = \overline a_{t-1}, \overline x_{t-1}).
\end{equation}
Let $m_{T+1} = Y$ and $w_{T+1} = w_{T+1}^\prime = 1$, and recursively define
\begin{equation} \label{eq:seq-reg}
    m_t(\overline X_t) = \bbE \left[ m_{t+1}(\overline X_{t+1}) w_{t+1}\{ \pi_{t+1}(\overline X_{t+1}) \} w_{t+1}^\prime \{ \pi_{t+1}^\prime( \overline X_{t+1}) \}  \,\, \bigg| \,\, \overline A_t = \overline a_t, \overline X_t \right]
\end{equation}
as the sequential regression function for $t < T+1$.  This recursive structure captures the dynamic nature of the problem: each $m_t$ represents the expected value of future weighted outcomes given the treatment and covariate history up to time $t$, weighted by the cross-world propensity scores. This is reminiscent of the weighted sequential regression function in \citet{schomaker2024causal}, but our construction weights on two sets of propensity scores simultaneously. Note that $m_0 = \psi(\overline a_T)$.  

\bigskip

Finally, let $r_t(A_t, \overline X_t) = \frac{\one(A_t = a_t) w_t\{ \pi_t(\overline X_t) \} w_t^\prime\{ \pi_t^\prime(\overline X_t) \}}{\pi_t(\overline X_t)}$ denote the inverse propensity score weight, and $r_t^\prime(A_t, \overline X_t) = \frac{\one(A_t = a_t^\prime) w_t\{ \pi_t(\overline X_t) \} w_t^\prime \{ \pi_t^\prime(\overline X_t) \}}{\pi_t^\prime (\overline X_t)}$ denote an augmented weight that incorporates the cross-world regime.

\subsection{Efficient influence function}

The efficient influence function is an important concept from semiparametric efficiency theory \citep{van2000asymptotic, tsiatis2006semiparametric}. It can be thought of as the first derivative in the von Mises expansion of the functional $\psi(\overline a_T)$ \citep{von1947asymptotic}. Practically, it is useful because it characterizes one way to debias standard estimators to attain doubly robust-style estimators with improved bias properties under model misspecification \citep{kennedy2024semiparametric}. The estimator we propose is based on the efficient influence function of $\psi(\overline a_T)$, whose form we derive in the following result. 
\begin{lemma} \label{lem:eif}
    Suppose Assumption~\ref{asmp:partial-common-support} holds, and 
    \begin{itemize}
        \item $w_t\{ \pi_t(\overline X_t) \}$ and $w_t^\prime\{ \pi_t^\prime(\overline X_t) \}$ are constructed such that $r_t(A_t; \overline X_t)$ and $r_t^\prime(A_t; \overline X_t)$ are uniformly bounded,  
        \item there exists $C > 0$ such that $\bbP \{ m_t(\overline X_t) < C \} = 1$ for all $t \leq T$, and 
        \item there exists $C^\prime > 0, \delta > 0$ such that $\sup_{t \leq T} \bbE \{ \rho_t^{2 + \delta}(\overline X_t) \mid \rho_t(\overline X_t) < \infty \} < C^\prime$.
    \end{itemize} 
    Then the uncentered efficient influence function of $\psi(\overline a_T)$ in a nonparametric model is
    \begin{align*}
        \varphi(Z) &= \varphi_m(Z) + \varphi_w(Z) \text{ where} \\
        \varphi_m(Z) &= m_1(X_1) w_1\{ \pi_1(X_1) \} w_1^\prime\{ \pi_1^\prime(X_1) \} \\
        &+ \sum_{t=1}^{T} \left\{ \prod_{s=1}^{t} r_s(A_s, \overline X_s) \right\} \left[ m_{t+1}(\overline X_{t+1})w_{t+1}\{ \pi_{t+1}(\overline X_{t+1}) \} w_{t+1}^\prime \{ \pi_{t+1}(\overline X_{t+1}) \} - m_t(\overline X_t) \right] \text{ and } \\
        \varphi_w(Z) &= \sum_{t=1}^{T} \left\{ \prod_{s=1}^{t-1} r_s(A_s, \overline X_s) \right\}  m_t(\overline X_t)\phi_t(A_t,\overline X_t)  w_t^\prime\{ \pi_t^\prime(\overline X_t) \}  \\ 
        &+\sum_{t=1}^T \left\{ \prod_{s=1}^{t-1} r_s^\prime(A_s, \overline X_s) \rho_s(\overline X_s) \right\} m_t(\overline X_t) w_t\{ \pi_t(\overline X_t) \} \rho_t(\overline X_t) \phi_t^\prime(A_t, \overline X_t),
    \end{align*}
    where
    \begin{align*}
        \phi_t(A_t, \overline X_t) &= \dot w_t \{ \pi_t(\overline X_t) \} \{ \one(A_t = a_t) - \pi_t(\overline X_t) \} \text{ and } \\
        \phi_t^\prime(A_t, \overline X_t) &= \dot w_t^\prime \{ \pi_t^\prime(\overline X_t) \} \{ \one(A_t = a_t^\prime) - \pi_t^\prime(\overline X_t) \},
    \end{align*}
    and $\dot w(x)$ denotes the first derivative of $w(x)$.
\end{lemma}

Before examining this result in detail, we discuss the assumptions. These assumptions ensure the efficient influence function is well-defined and has finite variance. First, the result assumes the partial common support assumption; otherwise $\psi(\overline a_T) = 0$ and the efficient influence function is also zero. When the partial common support assumption holds, the existence of the efficient influence function requires bounded variance of $\varphi(Z)$. To ensure this, Lemma~\ref{lem:eif} assumes that $w_t$ and $w_t^\prime$ are constructed so that $r_t$ and $r_t^\prime$ are uniformly bounded. This can be guaranteed through appropriate construction of the weights---overlap weights and versions of smooth trimming weights satisfy this condition. For example, $f(x) = 1 - \exp(-kx)$ for $k > 0$ satisfies this condition. Then, the result assumes that the sequential regression functions are uniformly bounded, which is a mild assumption. Finally, the result assumes that the $2+\delta$ moment of $\rho_t(\overline X_t)$ is bounded on the set where $\rho_t$ is finite. While this rules out heavy-tailed covariate densities, it still permits arbitrarily large density ratios.

\bigskip

The efficient influence function in Lemma~\ref{lem:eif} follows the typical general structure but with some important nuances. As usual, $\varphi(Z)$ consists of a plug-in estimator minus the true functional, plus weighted residual terms. The first component, $\varphi_m(Z)$, represents the efficient influence function that would arise if $\pi_t(\overline X_t)$ and $\pi_t^\prime(\overline X_t)$ were known and did not require estimation. The second component, $\varphi_w(Z)$, emerges from the necessity of estimating these quantities. This structure is typical of dynamic stochastic interventions and S-LMTPs where the intervention itself depends on unknown propensity scores.  

\bigskip

\noindent Beyond this general structure, two further aspects are worth emphasizing:
\begin{enumerate}
    \item \textbf{Dependence on the non-target regime.} Lemma~\ref{lem:eif} introduces a novel form of $\varphi_w(Z)$ that reflects the dependency on $w_t^\prime\{ \pi_t^\prime(\overline X_t) \}$ and $\pi_t^\prime(\overline X_t)$, which yields the second summand. This second summand is unusual because it addresses estimation with respect to the non-target regime, $\overline a_T^\prime$. Notice that both $r^\prime$ and $\phi_t^\prime$ include indicators for $\one(A_t = a_t^\prime)$, while $\rho_t(\overline X_t) = \frac{d\mathbb{P}(X_t \mid \overline A_{t-1} = \overline a_{t-1}, \overline X_{t-1})}{d\mathbb{P}(X_t \mid \overline A_{t-1} = \overline a_{t-1}^\prime, \overline X_{t-1})}$ transports from the non-target regime $\overline a_{t-1}^\prime$ to the target regime $\overline a_{t-1}$.
    \item \textbf{Deterministic and stochastic components.} The efficient influence function combines elements from both deterministic and stochastic interventions. This dual nature reflects the fact that while $w_t\{ \pi_t(\overline X_t) \} w_t^\prime\{ \pi_t^\prime(\overline X_t) \}$ act as a weight which may not equal zero or one---similar to the intervention propensity score in a stochastic intervention---the target effect corresponds to a deterministic intervention where $\overline A_T = \overline a_T$. Both the deterministic and stochastic characteristics appears in the weights $r_t(A_t, \overline X_t)$ and $r_t^\prime(A_t, \overline X_t)$. The deterministic component manifests through the indicator functions $\one(A_t = a_t)$, reminiscent of the classic doubly robust efficient influence function for the mean under a deterministic treatment regime \citep{bang2005doubly}. Meanwhile, the stochastic intervention aspect appears in $\tfrac{w_t\{ \pi_t(\overline X_t) \} w_t^\prime\{ \pi_t^\prime(\overline X_t) \}}{\pi_t(\overline X_t)}$ and $\tfrac{w_t\{ \pi_t(\overline X_t) \} w_t^\prime\{ \pi_t^\prime(\overline X_t) \}}{\pi_t^\prime(\overline X_t)}$.
\end{enumerate}

\subsection{Doubly robust-style estimator and convergence guarantees}

The efficient influence function inspires a doubly robust-style estimator. We assume access to $2n$ observations and use a sample split estimator, and let $\bbP_n$ denote the empirical mean over the estimation sample.  

\begin{algorithm}[Doubly robust-style estimator] \label{alg:dr-est} 
    Assume training and evaluation datasets of $n$ observations. 
    \begin{enumerate}
        \item For all timepoints, regress $A_t$ on $H_t = \{ \overline A_{t-1}, \overline X_t \}$ in the training data and obtain propensity score models; using these models, compute the weights $w_t\{ \widehat \pi_t(\overline X_t) \}$ and $w_t^\prime \{ \widehat \pi_t^\prime(\overline X_t) \}$, ratios $\widehat r_t(A_t, \overline X_t)$ and $\widehat r_t^\prime(A_t, \overline X_t)$, and efficient influence functions $\widehat \phi_t(A_t, \overline X_t)$ and $\widehat \phi_t^\prime(A_t, \overline X_t)$ on all samples and timepoints.
        \item For all timepoints, in the training sample construct an estimator for the covariate density ratio $\rho_t$ and compute $\widehat \rho_t(\overline X_t)$ in the estimation sample.
        \item For $t=T$ to $t=1$: 
        \begin{enumerate}
            \item 
            \begin{enumerate}
                \item If $t = T$, then $\widehat P_{T+1}(\overline X_{t+1}) = Y$. Otherwise, pseudo-outcome $\widehat P_{t+1}(\overline X_{t+1})$ is available from the previous step in loop (see step 3(b) below).
                \item In the training data, regress $\widehat P_{t+1}(\overline X_{t+1})$ against $A_t, H_t$ and plug in $\overline A_t = \overline a_t$ to obtain the sequential regression model $m_t(\overline X_t)$.
            \end{enumerate}
            \item Across the full data, compute pseudo-outcomes \\ $\widehat P_t(\overline X_t) = m_t(\overline X_t) w_t\{ \pi_t(\overline X_t) \} w_t^\prime\{ \pi_t^\prime(\overline X_t) \}$ to use in the next step.
        \end{enumerate}
        \item In the evaluation data only, compute the uncentered efficient function $\widehat \varphi(Z)$ by plugging nuisance estimates into $\varphi(Z)$ in Lemma~\ref{lem:eif}.
    \end{enumerate}
    Finally, output the point estimate and variance estimate
    \[
    \widehat \psi(\overline a_T) := \bbP_n \{ \widehat \varphi(Z) \} \text{ and } \widehat \sigma^2 := \bbP_n \left[ \{ \widehat \varphi(Z) - \widehat \psi(\overline a_T) \}^2 \right].
    \]
\end{algorithm}

Several features of this algorithm merit emphasis. First, the backward iteration in step 3 mirrors the sequential nature of the problem; this recursive structure is essential for capturing the dynamic dependencies inherent in longitudinal data. Second, the algorithm requires estimation of both propensity scores for the target and non-target regimes and the covariate density ratio $\rho_t$, reflecting the dependence on both regimes. We discuss estimating the covariate density ratio further in Section~\ref{subsec:ratio-estimation}, below. Finally, sample-splitting ensures that the same data are not used for both nuisance parameter estimation and the final estimator, which allows for asymptotic convergence guarantees without relying on Donsker or other conditions on the nuisance functions and their estimators (see, e.g., \citet{kennedy2024semiparametric}). To retain full-sample efficiency, one can cycle the folds, repeat the estimator, and average.  One can also construct an estimator based on more splits; five and ten folds are common.

\bigskip

Next, we establish $\sqrt{n}$-consistency and asymptotic normality under nonparametric conditions on the nuisance estimators.  
\begin{theorem} \label{thm:convergence}
    Under the setup of Lemma~\ref{lem:eif}, let $\widehat \psi(\overline a_T)$ and $\widehat \sigma^2$ denote a point estimate and variance estimate from Algorithm~\ref{alg:dr-est}, and let \\ $\widetilde m_t(\overline X_t) = \bbE \left[  \widehat m_{t+1}(\overline X_{t+1}) w_{t+1}\{ \widehat \pi_{t+1} (\overline X_{t+1}) \} w_{t+1}^\prime \{ \widehat \pi_{t+1}^\prime (\overline X_{t+1}) \} \mid \overline A_t = \overline a_t, \overline X_t \right]$ denote the sequential regression with estimated pseudo-outcome. Suppose
    \begin{enumerate}
        \item there exists $C < \infty$ such that $\bbP \big\{ \widehat m_t(\overline X_t) \leq C \big\} = 1$ for all $t \leq T$, \label{cond:bounded-seq-reg}     
        \item there exist $C^\prime > 0, \delta > 0$ such that $\sup_{t \leq T} \bbE \{ \widehat \rho^{2+\delta}(\overline X_t) \mid \widehat \rho_t(\overline X_t) < \infty \} < C^{\prime}$, \label{cond:bounded-moment-est}
        \item $\| \widehat \varphi - \varphi \| = o_\bbP(1)$, and \label{cond:consistent-eif}
        \item $\| \widehat \pi_t - \pi_t \| = o_\bbP(n^{-1/4}), \| \widehat \pi_t^\prime - \pi_t^\prime \| = o_\bbP(n^{-1/4}),   \| \widehat \pi_t^\prime - \pi_t^\prime \|_4 = o_\bbP(n^{-1/4}), \| \widehat \rho_t - \rho_t \|_4 = o_\bbP(n^{-1/4})$, and $\| \widehat m_t - \widetilde m_t \| = o_\bbP(n^{-1/4})$ for all $t \leq T$. \label{cond:bias}
    \end{enumerate}    
    Then,
    \[
    \sqrt{\frac{n}{\widehat \sigma^2}} \left\{ \widehat \psi(\overline a_T) - \psi(\overline a_T) \right\} \indist N \left( 0, 1 \right).
    \]    
\end{theorem}

The convergence result in Theorem~\ref{thm:convergence} provides important theoretical guarantees for the proposed estimator. The assumptions required are mild. Condition~\ref{cond:bounded-seq-reg} requires uniform boundedness of the outcome regression. Condition~\ref{cond:bounded-moment-est} asserts that the estimated covariate density ratios satisfy a bounded moment condition, as was assumed for the true ratios in Lemma~\ref{lem:eif}. Condition~\ref{cond:consistent-eif} asserts that the estimated influence function converges to the true influence function, which naturally follows from the consistency of the individual nuisance parameter estimates. Most importantly, condition~\ref{cond:bias} requires only $n^{-1/4}$ convergence rates for the nuisance parameters---a rate that is achievable even with flexible machine learning methods under standard assumptions \citep{gyorfi2002distribution}.  

\bigskip

The mild rate requirement is particularly noteworthy as it reflects the doubly robust nature of the estimator. The bias of the estimator $\widehat \psi(\overline a_T)$ can be bounded by a sum of products of errors of the nuisance functions. Therefore, when each individual nuisance parameter converges at the slower $n^{-1/4}$ rate, their products in the bias converge at the faster $n^{-1/2}$ rate, which is sufficient to maintain $\sqrt{n}$-consistency and asymptotic normality of the overall estimator. The $L_4(\bbP)$ convergence rates for the propensity scores $\pi_t^\prime$ and covariate density ratios $\rho_t$ arise from the structure of the bias terms, which involve products of covariate density ratios for which we only have bounded moment assumptions. The simplest bound we can obtain involves applying the Cauchy-Schwarz inequality twice, yielding $L_4(\bbP)$ requirements. These have appeared elsewhere in the causal inference literature; e.g., \citet{foster2023orthogonal}. A more refined analysis or alternative stronger assumptions could relax these conditions: uniform bounds on the covariate density ratios would permit the typical $L_2(\bbP)$ convergence rates, while supremum norm convergence for either nuisance function would allow the other to achieve only an $L_2(\bbP)$ rate. For ease of exposition, we focus on $L_4(\bbP)$ bounds. Moreover, also for ease of exposition, we state condition~\ref{cond:bias} as if every nuisance estimator must achieve the $n^{-1/4}$ rate; in practice, it suffices that each product of error terms attains the faster $n^{-1/2}$ rate, allowing some nuisances to converge more slowly provided their counterparts converge more quickly.

\subsection{Covariate density ratio estimation} \label{subsec:ratio-estimation}

We estimate the covariate density ratio $\rho_t(\overline X_t)$ using binary regression rather than direct density ratio estimation. This approach is motivated by the fact that direct density ratio estimators struggle in high-dimensional settings, while regression-based methods are more robust and widely available \citep{sugiyama2012density, choi2022density}.
We exploit the identity
\[
\rho_t(\overline X_t) \equiv \frac{d\bbP(X_t \mid \overline A_{t-1} = a_{t-1}, \overline X_{t-1})}{d\bbP(X_t \mid \overline A_{t-1} = \overline a_{t-1}^\prime, \overline X_{t-1})} = \frac{\bbP(\overline A_{t-1} = \overline a_{t-1} \mid \overline X_t) \bbP(\overline A_{t-1} = \overline a_{t-1}^\prime \mid \overline X_{t-1})}{\bbP(\overline A_{t-1} = \overline a_{t-1}^\prime \mid \overline X_t) \bbP(\overline A_{t-1} = \overline a_{t-1} \mid \overline X_{t-1})},
\]
which holds when $\rho_t(\overline X_t) < \infty$ almost surely---a condition guaranteed by convention that the propensity scores are zero whenever the conditioning set is undefined. Each of the four probabilities on the right-hand side can be estimated using standard nonparametric binary regression methods, making this approach both practical and well-supported by existing software.

\section{Data analysis} \label{sec:data-analysis}

In this section, we provide illustrative results from a data analysis examining the effect of union membership on wages. Our code is available at \url{https://github.com/alecmcclean/longitudinal-weighting}.

\bigskip

\subsection{Data Description}

We use the \texttt{wagepan} dataset from the \texttt{wooldridge} package in \texttt{R} \citep{wooldridge2024data, r2024language}. This dataset is from \citet{vella1998whose} and was obtained from the \textit{Journal of Applied Econometrics} archive at \url{http://qed.econ.queensu.ca/jae/}. The dataset contains employment information on 545 workers over eight years, from 1980-1987, though we focus on the first four years (1980-1983) for our analysis.

\bigskip

The dataset includes baseline covariates measured for each person: years of education, race (black/white), and ethnicity (hispanic/not hispanic). Additionally, several time-varying covariates are recorded for each year, including marital status (married/not married), health status (poor health: yes/no), labor market experience (in years), number of hours worked, occupation and industry classifications, region of residence (South/non-South), union membership, and the natural logarithm of hourly wage. An individual identifier allows us to link observations over time for each worker.

\subsection{Methodology}

Following \citet{vella1998whose}, we treat union membership as a time-varying treatment variable. Our two target regimes are \emph{always treated} ($\overline a_t = (1,1,1,1)$, i.e., always in a union) \emph{never treated} ($\overline a_t^\prime = (0,0,0,0)$, i.e., never in a union). We employ smooth trimming weights for both regimes:
\[
w_t(\pi_t) = 1 - \exp(-20 \pi_t) \text{ and } w_t^\prime(\pi_t^\prime) = 1 - \exp(-20 \pi_t^\prime).
\]
Our primary outcome of interest is the log wage in 1983. We estimated the cumulative cross-world weighted treatment effect, as in \eqref{eq:cumulative-weight}. To estimate both pieces of this effect, we use the doubly robust-style estimator in Algorithm~\ref{alg:dr-est} with five-fold cross-fitting. For nuisance function estimation, we employ the SuperLearner, combining multiple machine learning approaches: a linear model, generalized linear model, lasso with no interactions, a regression tree, and a random forest with default settings \citep{polley2024super, wright2017fast, therneau2023rpart}.

\subsection{Results}

Figure~\ref{fig:prop-scores} displays the distributions of the propensity scores under each of the targeted regimes. Notably, some probabilities are near zero. In the first row of Figure~\ref{fig:prop-scores}, propensity scores near zero indicate that certain workers would have close to zero probability of being in a union if they were in a union previously. In the second row, this indicates that certain workers would have almost probability one of being in a union if they were not in a union previously. This is a reasonable scenario for considering weighting to adapt to positivity violations. Table~\ref{tab:combined_summary} presents the distribution of the weights and the effective sample size at each timepoint. The effective sample size at each timepoint was calculated as the squared sum of cumulative weights up to that timepoint divided by the sum of squares of cumulative weights. Together, these results show that the weights were not extreme.  Table~\ref{tab:combined_summary} also provides evidence on the common support assumption, in Assumption~\ref{asmp:partial-common-support}. It shows that the covariate density ratios for timepoints 2 to 4 are not that extreme, suggesting that the common support assumption may hold for all subjects.

\bigskip

\begin{table}[ht]
    \centering
    \begin{tabular}{lcccccc}
        \toprule
        \textbf{Statistic} & $\prod_{t=1}^{4} w_t(\pi_t)\,w_t'(\pi_t')$ & $\rho_{2}$ & $\rho_{3}$ & $\rho_{4}$ & \textbf{Time} & \textbf{ESS} \\
        \midrule
        Min\phantom{.}                     & 0.161   & 0.144   & 0.121   & 0.212   & 1 & 539 \\
        1\textsuperscript{st} Quartile    & 0.486   & 0.761   & 0.736   & 0.816   & 2 & 535 \\
        Median                             & 0.630   & 1.060   & 0.990   & 1.037   & 3 & 528 \\
        Mean                               & 0.620   & 1.155   & 1.068   & 1.152   & 4 & 505 \\
        3\textsuperscript{rd} Quartile    & 0.757   & 1.426   & 1.321   & 1.296   &   &     \\
        Max                                & 0.959   & 4.126   & 3.241   & 11.80  &   &     \\
        \bottomrule
    \end{tabular}
    \caption{Combined summary of the distribution of cumulative weights (column 2), the covariate density ratio distributions at times 2–4 (columns 3-5), and the effective sample size (ESS) at each timepoint.}
    \label{tab:combined_summary}
\end{table}

Table~\ref{tab:results} presents the final results. We reject the null hypothesis of no effect at $\alpha = 0.05$, with an estimated effect of 0.216 (95\% CI: [0.046, 0.387]). Interpreting this on the log wage scale, this corresponds to approximately a 22\% wage increase in 1983. These results suggest union membership had a positive causal effect on wages in the early 1980s. However, these findings rely on the assumption of no unmeasured confounding, which may be questionable in this observational setting, warranting sensitivity analyses in future work.

\begin{table}[ht]
\centering
\begin{tabular}{l r}
\toprule
\textbf{Metric} & \textbf{Value} \\
\midrule
Cumulative cross-world weighted treatment effect      & 0.216 \\
95\% CI       & [0.046,\,0.387] \\
\bottomrule
\end{tabular}
\caption{Point estimate and 95\% confidence interval.}
\label{tab:results}
\end{table}

\begin{figure}[ht]
    \centering
    \includegraphics[width = 0.9\linewidth]{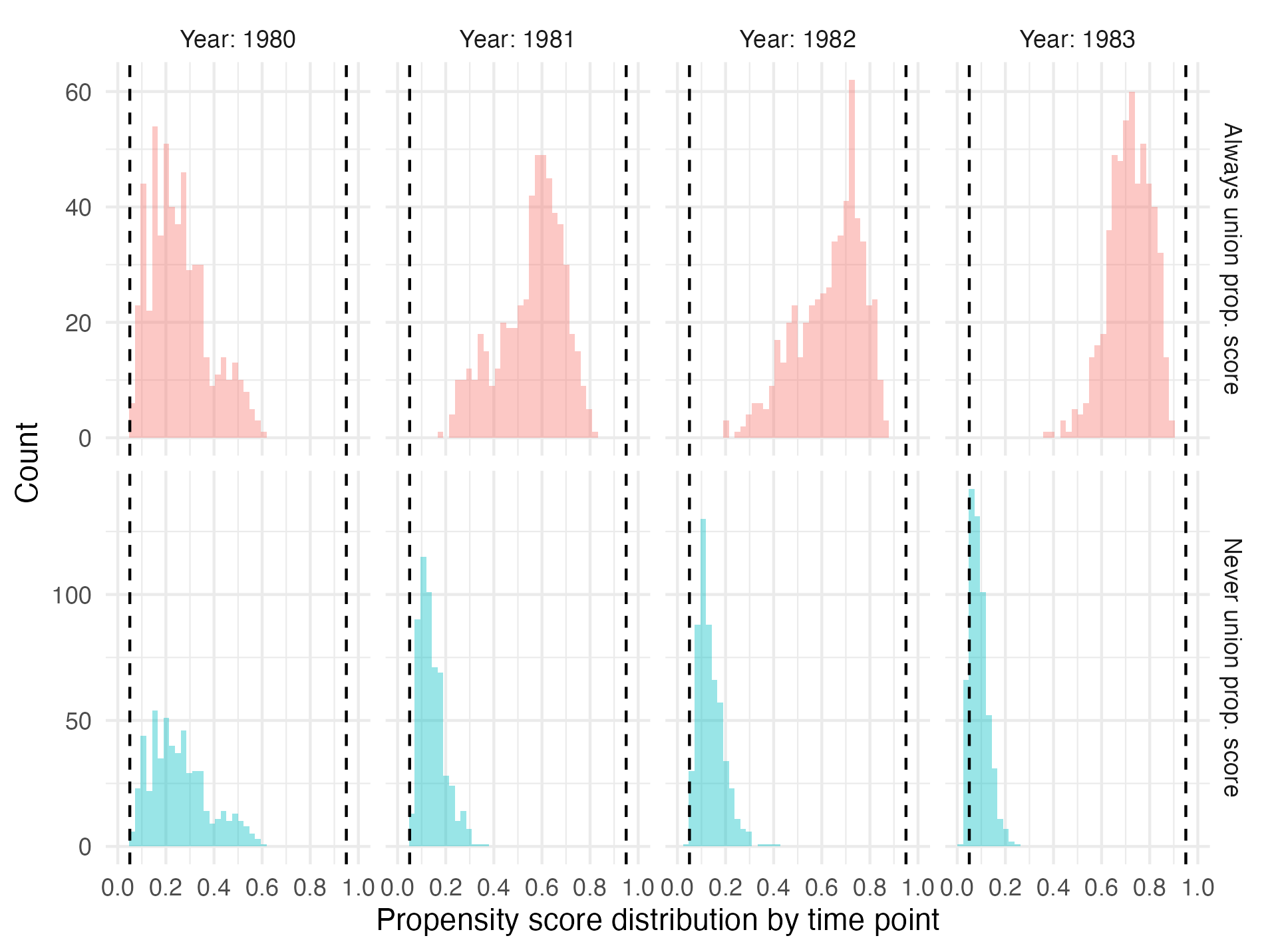}
    \caption{Propensity score distribution.}
    \label{fig:prop-scores}
\end{figure}

\section{Discussion} \label{sec:discussion}

This paper formalizes cumulative cross-world weighted effects, which address positivity violations in longitudinal settings while preserving the scientific target of interest—the mean difference in potential outcomes under two treatment regimes. We establish identification under strong sequential randomization and develop efficient estimation methods through doubly robust approaches.

\bigskip

A key contribution is explicitly highlighting a fundamental trade-off in longitudinal causal inference. Cross-world effects isolate the pure causal contrast between regimes, offering the clearest answer to mechanistic questions. However, these effects cannot be implemented as feasible interventions, similar to natural direct effects in mediation or survivor effects subject to censoring by death \citep{frangakis2002principal}. In contrast, stochastic longitudinal modified treatment policies, such as flip interventions, are implementable but yield effects that conflate the intervention’s impact on the ultimate outcome with its impact on intermediate covariates and treatments. Thus, the preferred approach depends on whether a given research question prioritizes mechanistic understanding or practical policy implementation.

\bigskip

Another conceptual contribution is our demonstration of the importance of common support for time-varying covariate distributions, which ensures the identified weighted functional faithfully represents genuine underlying causal contrasts. While violations of common support do not render the mechanistic causal contrast meaningless, they do indicate that the identified functional estimated from observed data no longer accurately reflects this contrast. Therefore, practitioners should explicitly assess common support using covariate density ratios, as illustrated in our data analysis.

\bigskip

Future work could explore data-adaptive methods for selecting weights and extend beyond propensity score weighting towards direct covariate balancing approaches. Additionally, sensitivity analyses specifically tailored to cross-world estimands merit further investigation.

\section*{References}
\vspace{-0.3in}
\bibliographystyle{plainnat}
\bibliography{references}

\newpage
\appendix

\newgeometry{margin=1in} 
\color{black}
\section*{Appendix}

\begin{itemize}
    \item[Appendix~\ref{app:id}] contains the proofs for the identification results in Section~\ref{sec:identification}.
    \item[Appendix~\ref{app:estimation}] contains the proofs for the efficiency theory results in Section~\ref{sec:estimation}  
\end{itemize}

\section{Proofs for Section~\ref{sec:identification}} \label{app:id}

\subsection{Lemma~\ref{lem:id-prop-scores}}

\begin{proof}
    We'll focus only on identifying $\bbP \{ A_t (\overline a_{t-1}) = a_t \mid \overline X_t(\overline a_{t-1}) \}$ because the argument for $\bbP \{ A_t (\overline a_{t-1}^\prime) = a_t^\prime \mid \overline X_t(\overline a_{t-1}^\prime) \}$ follows trivially.  We have
    \begin{align*}
        \bbP \{ A_t(\overline a_{t-1}) = a_t \mid \overline X_t(\overline a_{t-1})  \} &= \frac{d\bbP \{ A_t(\overline a_{t-1}) = a_t, X_t(\overline a_{t-1}), \dots, X_2(a_1) \mid X_1 \}}{d\bbP \{ X_t(\overline a_{t-1}), X_2(a_1) \mid X_1 \}} \\
        &= \frac{\bbP \{ A_t(\overline a_{t-1}) = a_t, X_t(\overline a_{t-1}), \dots, X_2 \mid A_1 = a_1 X_1 \}}{d\bbP \{ X_t(\overline a_{t-1}), \dots, X_2 \mid A_1 = a_1, X_1 \}} \\ 
        &= \frac{d\bbP \{ A_t(\overline a_{t-1}) = a_t, X_t(\overline a_{t-1}), X_3(\overline a_2) \mid X_2, A_1 = a_1, X_1 \}}{d\bbP \{ X_t(\overline a_{t-1}), \dots, X_3(\overline a_2) \mid X_2, A_1 = a_1, X_1 \}} \\
        &= \frac{d\bbP \{ A_t(\overline a_{t-1}) = a_t, X_t(\overline a_{t-1}), X_3 \mid  A_2 = a_2, X_2, A_1 = a_1 X_1 \}}{d\bbP \{ X_t(\overline a_{t-1}), \dots, X_3 \mid A_2 = a_2, X_2, A_1 = a_1, X_1 \}}
    \end{align*}
    where the first line follows by factoring the conditional probability and the second line by strong sequential randomization, the positivity assumption, that $\pi_1(X_1) > 0$, so that the conditioning set is well-defined, and by consistency, so that $X_2(a_1) = X_2$, conditional on $A_1 = a_1, X_1$. We leave the other left-hand side variables as their natural values for ease of notation, although one could write $X_3(a_2, A_1)$ and so on, instead of what we have written. The third line follows by factoring the conditional probabilities again, and the fourth line follows by strong sequential randomization, positivity, and consistency. This argument can be repeated $t-3$ more times to yield
    \[
    \bbP \{ A_t(\overline a_{t-1}) = a_t \mid \overline X_t(\overline a_{t-1})  \} = \frac{d\bbP ( A_t = a_t, X_t \mid \overline A_{t-1} = \overline a_{t-1}, \overline X_{t-1} )}{d\bbP ( X_t \mid \overline A_{t-1} = \overline a_{t-1}, \overline X_{t-1} )} = \bbP(A_t = a_t \mid \overline X_t, \overline A_{t-1} = \overline a_{t-1}),
    \]
    where the second equality follows by Bayes' rule.
\end{proof}

\subsection{Theorem~\ref{thm:id-cumulative}}

\begin{proof}
    \textbf{Timepoint 1: } We begin with \footnotesize
    \begin{align*}
        &\bbE \left( Y(\overline a_T) \prod_{t=1}^T w_t\big[ \bbP \{ A_t(\overline a_{t-1}) = a_t \mid \overline X_t(\overline a_{t-1}) \} \big] w_t^\prime \big[ \bbP \{ A_t(\overline a_{t-1}^\prime) = a_t \mid \overline X_t(\overline a_{t-1}^\prime) \} \big] \right) \\
        &= \bbE \left( \bbE \left[ Y(\overline a_T) \prod_{t=2}^T w_t\big[ \bbP \{ A_t(\overline a_{t-1}) = a_t \mid \overline X_t(\overline a_{t-1}) \} \big] w_t^\prime \big[ \bbP \{ A_t(\overline a_{t-1}^\prime) = a_t \mid \overline X_t(\overline a_{t-1}^\prime) \} \big] \mid X_1 \right] w_1\{ \pi_1(X_1)\} w_1^\prime \{ \pi_1^\prime(X_1) \} \right),
    \end{align*} \normalsize
    which follows by iterated expectations on $X_1$.  Since $\pi_1(X_1) \pi_1^\prime(X_1) = 0 \implies w_1\{ \pi_1(X_1) \} w_1^\prime\{ \pi_1^\prime(X_1)\} = 0$ by assumption, in the inner expectation is equal to the expectation that also conditions on the event $\big\{ \pi_1(X_1) > 0, \pi_1^\prime(X_1) > 0 \big\}$, because otherwise the overall expectation is zero.
    
    \bigskip
    
    \noindent \textbf{Timepoint 2:} Then, by Lemma~\ref{lem:id-prop-scores}, the second timepoint propensity scores can be identified. Revisiting the overall g-formula, and omitting some arguments for brevity, we have \footnotesize
    \begin{align*}
        &\bbE \left( \bbE \left[ Y(\overline a_T) \prod_{t=2}^T w_t\big[ \bbP \{ A_t(\overline a_{t-1}) = a_t \mid \overline X_t(\overline a_{t-1}) \} \big] w_t^\prime \big[ \bbP \{ A_t(\overline a_{t-1}^\prime) = a_t \mid \overline X_t(\overline a_{t-1}^\prime) \} \big] \mid X_1 \right] w_1 ( \pi_1 ) w_1^\prime (\pi_1^\prime) \right) \\
        &= \bbE \left( \bbE \left[ Y(\overline a_T) \prod_{t=3}^T w_t\big[ \bbP \{ A_t(\overline a_{t-1}) = a_t \mid \overline X_t(\overline a_{t-1}) \} \big] w_t^\prime \big[ \bbP \{ A_t(\overline a_{t-1}^\prime) = a_t \mid \overline X_t(\overline a_{t-1}^\prime) \} \big] w_2(\pi_2) w_2^\prime(\pi_2^\prime) \mid X_1 \right] w_1 ( \pi_1 ) w_1^\prime (\pi_1^\prime) \right)
    \end{align*} \normalsize
    Then, by strong sequential exchangeability and iterated expectations on $X_2 \mid A_1 = a_1, X_1$, we have \footnotesize
    \begin{align*}
        &\bbE \left( \bbE \left[ Y(\overline a_T) \prod_{t=3}^T w_t\big[ \bbP \{ A_t(\overline a_{t-1}) = a_t \mid \overline X_t(\overline a_{t-1}) \} \big] w_t^\prime \big[ \bbP \{ A_t(\overline a_{t-1}^\prime) = a_t \mid \overline X_t(\overline a_{t-1}^\prime) \} \big] w_2(\pi_2) w_2^\prime(\pi_2^\prime) \mid X_1 \right] w_1 ( \pi_1 ) w_1^\prime (\pi_1^\prime) \right)  \\
        &\hspace{-0.75in}= \bbE \left\{ \bbE \left( \bbE \left[ Y(\overline a_T) \prod_{t=3}^T w_t\big[ \bbP \{ A_t(\overline a_{t-1}) = a_t \mid \overline X_t(\overline a_{t-1}) \} \big] w_t^\prime \big[ \bbP \{ A_t(\overline a_{t-1}^\prime) = a_t \mid \overline X_t(\overline a_{t-1}^\prime) \} \big] \mid X_2, A_1 = a_1, X_1 \right] w_2(\pi_2) w_2^\prime(\pi_2^\prime) \mid X_1 \right) w_1 ( \pi_1 ) w_1^\prime (\pi_1^\prime) \right\}.
    \end{align*}
    \normalsize
    Next, notice that $\pi_2 \pi_2^\prime = 0 \implies w_2 (\pi_2) w_2^\prime(\pi_2^\prime) = 0$. Hence, the innermost expectation is equal to the expectation that also conditions on the event 
    \[
    \left\{ \pi_1(X_1) > 0, \pi_1^\prime(X_1) > 0, \pi_2(\overline X_2) > 0, \pi_2^\prime(\overline X_2) > 0 \right\}.
    \]
    Moreover, note that when the arguments of the second timepoint propensity scores are undefined then the propensity scores equal zero by convention and the whole expectation is equal to zero. Therefore, either the inner expectation is well-defined or the overall expectation is zero.

    \bigskip

    \noindent \textbf{Induction: } The rest of the steps to obtain the full identification follow by repeating the same argument as for $t=2$. Identifying the potential outcomes follows by consistency.    
\end{proof}

\section{Proofs for Section~\ref{sec:estimation}} \label{app:estimation}

Throughout, we will omit arguments of functions for brevity unless necessary for clarity. As in the main paper, let $\widetilde m(\overline X_t) = \bbE ( \widehat m_t \widehat w_t \widehat w_t^\prime \mid \overline a_t, \overline X_t)$.  We will first establish a series of helper lemmas, which imply both a bound on the bias and the result for the efficient influence function.

\begin{lemma} \label{lem:eif-helper-1}
    Under the setup of Lemma~\ref{lem:eif}, 
    \begin{align*}
        \bbE \{ \widehat \varphi_m(Z) \} &= m_0 + \sum_{t=1}^{T}  \bbE \left\{ \left( \prod_{s=1}^{t} \widehat r_s - \prod_{s=1}^{t}  r_s \right) \left( \widetilde m_t - \widehat m_t \right) \right\} \\
        &+ \sum_{t=1}^{T} \bbE \left[ \left( \prod_{s=1}^{t-1}  r_s \right) \bbE \left\{ \widehat m_t \left( \widehat w_t \widehat w_t^\prime - w_t w_t^\prime \right) \mid \overline a_{t-1}, \overline X_{t-1} \right\} \right]
    \end{align*}
\end{lemma}

\begin{proof}
    \begin{align*}
        \bbE \{ \widehat \varphi_m(Z) \} &= \bbE \left\{ \widehat m_1\widehat w_1 \widehat w_1^\prime + \sum_{t=1}^{T} \left( \prod_{s=1}^{t} \widehat r_s \right) \left( \widehat m_{t+1} \widehat w_{t+1} \widehat w_{t+1}^\prime - \widehat m_t \right) \right\} \\
        &= \widetilde m_0 + \sum_{t=1}^{T} \bbE \left\{ \left( \prod_{s=1}^{t} \widehat r_s \right) \left( \widetilde m_t - \widehat m_t \right) \right\} \\
        &= \widetilde m_0 + \sum_{t=1}^T \bbE \left\{ \left( \prod_{s=1}^{t} \widehat r_s - \prod_{s=1}^{t} r_s \right) \left( \widetilde m_t - \widehat m_t \right) + \left( \prod_{s=1}^t r_s \right) (\widetilde m_t - \widehat m_t) \right\} 
    \end{align*}
    where second equality follows by iterated expectations on $\overline A_t = \overline a_t, \overline X_t$ within each summand and the third line follows by adding zero. Next, we have
    \begin{align*}
        \widetilde m_0 &+ \sum_{t=1}^T \bbE \left\{ \left( \prod_{s=1}^t r_s \right) ( \widetilde m_t - \widehat m_t) \right\} = \sum_{t=0}^{T-1} \bbE \left\{ \left(\prod_{s=0}^t r_s \right) \widetilde m_t \right\} - \sum_{t=0}^{T-1} \bbE \left\{ \left( \prod_{s=0}^{t+1} r_s \right) \widehat m_{t+1} \right\} + m_0 \\
        &= \sum_{t=0}^{T-1} \bbE \left[ \left( \prod_{s=0}^t r_s \right) \left\{ \bbE \left( \widehat m_{t+1} \widehat w_{t+1} \widehat w_{t+1}^\prime - \widehat m_{t+1} w_{t+1} w_{t+1}^\prime \mid \overline a_t, \overline X_t \right) \right\} \right] + m_0 \\
        &= \sum_{t=1}^{T} \bbE \left[ \left( \prod_{s=1}^{t-1} r_s \right) \bbE \left\{ \widehat m_t \left( \widehat w_t \widehat w_t^\prime - w_t w_t^\prime \right) \mid \overline a_{t-1} \overline X_{t-1} \right\} \right] + m_0 
    \end{align*}
    where the first line follows by rearranging the sum over timepoints and because $\bbE \left\{ \left( \prod_{s=1}^T r_s \right) Y \right\} = m_0$, the second by iterated expectations on $\overline a_t, \overline X_t$, and final line by gathering terms.
\end{proof}

\begin{lemma} \label{lem:eif-helper-2}
    Under the setup of Lemma~\ref{lem:eif},  
    \begin{align*}
        \bbE \{ \widehat \varphi_w(Z) \} &= \sum_{t=1}^{T} \bbE \Bigg\{ \left( \prod_{s=1}^{t-1} \widehat r_s - \prod_{s=1}^{t-1} r_s \right) \widehat m_t \bbE \left( \widehat \phi_t \mid \overline X_t, \overline a_{t-1} \right) \widehat w_t^\prime \Bigg\} \\
        &+ \sum_{t=1}^T \bbE \Bigg\{ \left( \prod_{s=1}^{t-1} \widehat r_s^\prime \widehat \rho_s - \prod_{s=1}^{t-1} r_s^\prime \rho_s \right) \widehat m_t \widehat w_t \widehat \rho_t \bbE \left( \widehat \phi_t^\prime \mid \overline a_{t-1}^\prime, \overline X_t \right) \Bigg\} \\
        &+ \sum_{t=1}^{T}\bbE \Bigg\{ \left( \prod_{s=1}^{t-1} r_s^\prime \rho_s \right) \widehat m_t \widehat w_t (\widehat \rho_t - \rho_t) \bbE \left( \widehat \phi_t^\prime \mid \overline a_{t-1}^\prime, \overline X_t \right) \Bigg\} \\
        &+ \sum_{t=1}^{T} \bbE \left[ \left( \prod_{s=1}^{t-1} r_s \right) \widehat m_t \left\{ \bbE \left( \widehat \phi_t \mid \overline X_t, \overline a_{t-1} \right) \widehat w_t^\prime + \widehat w_t \bbE \left( \widehat \phi_t^\prime \mid \overline a_{t-1}^\prime, \overline X_t \right) + \widehat w_t \widehat w_t^\prime - w_t w_t^\prime \right\} \right] \\
        &+ \sum_{t=1}^{T} \bbE \left\{ \left( \prod_{s=1}^{t-1} r_s \right) \widehat m_t \left( w_t w_t^\prime - \widehat w_t \widehat w_t^\prime \right) \right\}
    \end{align*}
\end{lemma}

\begin{proof}
    To start, 
    \begin{align*}
        \bbE \{ \widehat \varphi_w(Z) \} &= \sum_{t=1}^{T} \bbE \Bigg\{  \left( \prod_{s=1}^{t-1} \widehat r_s \right) \widehat m_t \widehat \phi_t \widehat w_t^\prime + \left( \prod_{s=1}^{t-1} \widehat r_s^\prime \widehat \rho_s  \right) \widehat m_t \widehat w_t \widehat \rho_t \widehat \phi_t^\prime \Bigg\} \\
        &= \sum_{t=1}^{T} \bbE \Bigg\{ \left( \prod_{s=1}^{t-1} \widehat r_s - \prod_{s=1}^{t-1} r_s \right) \widehat m_t \bbE \left( \widehat \phi_t \mid \overline X_t, \overline a_{t-1} \right) \widehat w_t^\prime + \left( \prod_{s=1}^{t-1} r_s \right) \widehat m_t \bbE \left( \widehat \phi_t \mid \overline X_t, \overline a_{t-1} \right) \widehat w_t^\prime \Bigg\} \\
        &+ \sum_{t=1}^T \bbE \Bigg\{ \left( \prod_{s=1}^{t-1} \widehat r_s^\prime \widehat \rho_s - \prod_{s=1}^{t-1} r_s^\prime \rho_s \right) \widehat m_t \widehat w_t \widehat \rho_t \bbE \left( \widehat \phi_t^\prime \mid \overline a_{t-1}^\prime, \overline X_t \right) + \left( \prod_{s=1}^{t-1} r_s^\prime \rho_s  \right) \widehat m_t \widehat w_t \widehat \rho_t \bbE \left( \widehat \phi_t^\prime \mid \overline a_{t-1}^\prime, \overline X_t \right) \Bigg\}
    \end{align*}
    where the second equality follows by adding zero and iterated expectations on $\overline A_{t-1} = \overline a_{t-1}, \overline X_t$ in the first line and iterated expectations on $\overline A_{t-1} = \overline a_{t-1}^\prime, \overline X_t$ in the second line. 

    \medskip

    The first and third summands at the bottom of the above display appear in the final result. Considering the final summand, we have
    \begin{align*}
        &\sum_{t=1}^{T} \bbE \Bigg\{ \left( \prod_{s=1}^{t-1} r_s^\prime \rho_s \right) \widehat m_t \widehat w_t \widehat \rho_t \bbE \left( \widehat \phi_t^\prime \mid \overline a_{t-1}^\prime, \overline X_t \right) \Bigg\} \\
        &= \sum_{t=1}^{T}\bbE \Bigg\{ \left( \prod_{s=1}^{t-1} r_s^\prime \rho_s \right) \widehat m_t \widehat w_t (\widehat \rho_t - \rho_t) \bbE \left( \widehat \phi_t^\prime \mid \overline a_{t-1}^\prime, \overline X_t \right) \Bigg\} + \sum_{t=1}^{T}\bbE \Bigg\{ \left( \prod_{s=1}^{t-1} r_s^\prime \rho_s \right) \widehat m_t \widehat w_t \rho_t \bbE \left( \widehat \phi_t^\prime \mid \overline a_{t-1}^\prime, \overline X_t \right) \Bigg\} 
    \end{align*}
    by adding zero. The first summand on the final line appears in the result. For the second summand, notice that
    \begin{align*}
        &\sum_{t=1}^{T} \bbE \left\{ \left( \prod_{s=1}^{t-1} r_s^\prime \rho_s \right) \widehat m_t \widehat w_t \rho_t \bbE \left( \widehat \phi_t^\prime \mid \overline a_{t-1}^\prime, \overline X_t \right) \right\} \\
        &= \sum_{t=1}^{T} \int_{\overline{\mathcal{X}}_t} \widehat m_t \widehat w_t \bbE \left( \widehat \phi_t^\prime \mid \overline a_{t-1}^\prime, \overline x_t \right) d\bbP(x_t \mid \overline a_{t-1}, \overline x_{t-1}) \prod_{s=1}^{t-1} w_s(\overline x_s) w_s^\prime(\overline x_s) d\bbP(x_s \mid \overline a_{s-1}, \overline x_s) \\
        &= \sum_{t=1}^{T} \bbE \left\{ \left( \prod_{s=1}^{t-1} r_s \right) \widehat m_t \widehat w_t \bbE \left( \widehat \phi_t^\prime \mid \overline a_{t-1}^\prime, \overline X_t \right) \right\}
    \end{align*}
    where the first line follows by repeating, from $s = t-1$ to $s=1$, iterated expectations on $\overline X_s, \overline A_{s-1} = \overline a_{s-1}$ and then iterated expectations on $\overline X_{s-1}, \overline A_{s-1} = \overline a_{s-1}^\prime$ and canceling terms, and the second line follows by repeating, from $s=1$ to $t-1$, iterated expectations on $\overline X_{s}, \overline A_s = \overline a_{s-1}$.

    \medskip

    Next, we revisit the outstanding terms:
    \begin{align*}
        &\sum_{t=1}^{T} \bbE \left\{ \left( \prod_{s=1}^{t-1} r_s \right) \widehat m_t \bbE \left( \widehat \phi_t \mid \overline X_t, \overline a_{t-1} \right) \widehat w_t^\prime \right\} + \bbE \left\{ \left( \prod_{s=1}^{t-1} r_s \right) \widehat m_t \widehat w_t \bbE \left( \widehat \phi_t^\prime \mid \overline a_{t-1}^\prime, \overline X_t \right) \right\} \\
        &= \sum_{t=1}^{T} \bbE \left[ \left( \prod_{s=1}^{t-1} r_s \right) \widehat m_t \left\{ \bbE \left( \widehat \phi_t \mid \overline X_t, \overline a_{t-1} \right) \widehat w_t^\prime + \widehat w_t \bbE \left( \widehat \phi_t^\prime \mid \overline a_{t-1}^\prime, \overline X_t \right) \right\} \right] \\
        &= \sum_{t=1}^{T} \bbE \left[ \left( \prod_{s=1}^{t-1} r_s \right) \widehat m_t \left\{ \bbE \left( \widehat \phi_t \mid \overline X_t, \overline a_{t-1} \right) \widehat w_t^\prime + \widehat w_t \bbE \left( \widehat \phi_t^\prime \mid \overline a_{t-1}^\prime, \overline X_t \right) + \widehat w_t \widehat w_t^\prime - w_t w_t^\prime \right\} \right] \\
        &+ \sum_{t=1}^{T} \bbE \left\{ \left( \prod_{s=1}^{t-1} r_s \right) \widehat m_t \left( w_t w_t^\prime - \widehat w_t \widehat w_t^\prime \right) \right\}
    \end{align*}
    where the final line follows by adding zero. Combining terms from the above argument yields the result. 
\end{proof}

\begin{lemma} \label{lem:eif-helper-3}
    Under the setup of Lemma~\ref{lem:eif}, 
    \begin{align*}
        \bbE \{ \widehat \varphi(Z) \} - \psi(\overline a_T) &= \sum_{t=1}^{T}  \bbE \left\{ \left( \prod_{s=1}^{t} \widehat r_s - \prod_{s=1}^{t}  r_s \right) \left( \widetilde m_t - \widehat m_t \right) \right\} \\
        &+ \sum_{t=1}^{T} \bbE \Bigg\{ \left( \prod_{s=1}^{t-1} \widehat r_s - \prod_{s=1}^{t-1} r_s \right) \widehat m_t \bbE \left( \widehat \phi_t \mid \overline X_t, \overline a_{t-1} \right) \widehat w_t^\prime \Bigg\} \\
        &+ \sum_{t=1}^T \bbE \Bigg\{ \left( \prod_{s=1}^{t-1} \widehat r_s^\prime \widehat \rho_s - \prod_{s=1}^{t-1} r_s^\prime \rho_s \right) \widehat m_t \widehat w_t \widehat \rho_t \bbE \left( \widehat \phi_t^\prime \mid \overline a_{t-1}^\prime, \overline X_t \right) \Bigg\} \\
        &+ \sum_{t=1}^{T}\bbE \Bigg\{ \left( \prod_{s=1}^{t-1} r_s^\prime \rho_s \right) \widehat m_t \widehat w_t (\widehat \rho_t - \rho_t) \bbE \left( \widehat \phi_t^\prime \mid \overline a_{t-1}^\prime, \overline X_t \right) \Bigg\} \\
        &+ \sum_{t=1}^{T} \bbE \left[ \left( \prod_{s=1}^{t-1} r_s \right) \widehat m_t \left\{ \bbE \left( \widehat \phi_t \mid \overline X_t, \overline a_{t-1} \right) \widehat w_t^\prime + \widehat w_t \bbE \left( \widehat \phi_t^\prime \mid \overline a_{t-1}^\prime, \overline X_t \right) + \widehat w_t \widehat w_t^\prime - w_t w_t^\prime \right\} \right].
    \end{align*}
\end{lemma}
\begin{proof}
    Notice that the final lines in the previous two lemmas cancel, and recall that $\varphi(Z) = \varphi_m(Z) + \varphi_w(Z)$.
\end{proof}

\begin{lemma} \label{lem:eif-helper-4}
    Under the setup of Lemma~\ref{lem:eif}, 
    \begin{align*}
        &\bbE \left\{ \widehat \phi_t(A_t, \overline X_t) \mid \overline X_t, \overline A_{t-1} = \overline a_{t-1}) \right\} + \bbE \left\{ \widehat \phi_t^\prime(A_t, \overline X_t) \mid \overline X_t, \overline A_{t-1} = \overline a_{t-1}^\prime \right\} \\
        &\hspace{0.5in}+ \widehat w_t\{ \pi_t(\overline X_t) \} \widehat w_t^\prime\{ \pi_t^\prime(\overline X_t) \}  - w_t\{ \pi_t(\overline X_t) \} w_t^\prime\{ \pi_t^\prime(\overline X_t) \} \\
        &= \dot w_t\{\widehat \pi(\overline X_t) \} \dot w_t \{ \widehat \pi^{\prime}(\overline X_t)\}\{ \widehat \pi(\overline X_t) - \pi(\overline X_t)\} \{ \pi^\prime(\overline X_t) - \widehat \pi^\prime(\overline X_t)\} \\
        &\hspace{0.5in}- \frac{1}{2} w_t\{\widehat \pi(\overline X_t)\}\ddot w_t\{\widehat \pi^\prime(\overline X_t)\} \{ \widehat \pi^\prime(\overline X_t) - \pi^\prime(\overline X_t)\}^2 - \frac{1}{2} w_t\{ \widehat \pi^\prime(\overline X_t) \} \ddot w_t\{ \widehat \pi(\overline X_t) \} \{ \widehat \pi(\overline X_t) - \pi(\overline X_t)\}^2 \\
        &\hspace{0.5in}+ o\big[ \{ \widehat \pi(\overline X_t) - \pi(\overline X_t)\}^2 \big] + o\big[ \{ \widehat \pi^\prime(\overline X_t) + \pi^\prime(\overline X_t)\}^2 \big]
    \end{align*}
    and
    \begin{align*}
        \bbE \left\{ \widehat \phi_t(A_t, \overline X_t) \mid \overline X_t, \overline A_{t-1} = \overline a_{t-1}) \right\} &= \dot w_t(\widehat \pi_t) ( \pi_t - \widehat \pi_t ) 
    \end{align*}
    where $\ddot w_t$ denotes the second derivative of $w_t(x)$.
\end{lemma}

\begin{proof}
    We have
    \begin{align*}
        &\bbE \left\{ \widehat \phi_t(A_t, \overline X_t) \mid \overline X_t, \overline A_{t-1} = \overline a_{t-1}) \right\} + \bbE \left\{ \widehat \phi_t^\prime(A_t, \overline X_t) \mid \overline X_t, \overline A_{t-1} = \overline a_{t-1}^\prime \right\} \\
        &\hspace{0.5in}+ \widehat w_t\{ \pi_t(\overline X_t) \} \widehat w_t^\prime\{ \pi_t^\prime(\overline X_t) \}  - w_t\{ \pi_t(\overline X_t) \} w_t^\prime\{ \pi_t^\prime(\overline X_t) \} \\
        &= \dot w_t(\widehat \pi)(\pi - \widehat \pi) w_t^\prime(\widehat \pi^\prime) + \dot w_t^\prime(\widehat \pi^\prime)(\pi^\prime - \widehat \pi^\prime) w_t(\widehat \pi) + w_t(\widehat \pi) w_t^\prime(\widehat \pi^\prime) - w_t(\pi) w_t(\pi^\prime).
    \end{align*}
    Second-order Taylor expansions of $w_t(\pi)$ and $w_t^\prime(\pi^\prime)$ around $w_t(\widehat \pi)$ and $w_t^\prime(\widehat \pi^\prime)$, respectively, yield
    \begin{align*}
        w_t(\pi) w_t^\prime(\pi^\prime) &= w_t(\widehat \pi)w_t^\prime(\widehat \pi^\prime) + \dot w_t(\widehat \pi^\prime) (\pi^\prime - \widehat \pi^\prime) w_t(\widehat \pi) + \dot w_t(\widehat \pi) (\pi - \widehat \pi) w_t^\prime(\widehat \pi^\prime) \\
        &+ \dot w_t(\widehat \pi)\dot w_t^\prime(\widehat \pi^{\prime})(\widehat \pi - \pi)(\widehat \pi^\prime - \pi^\prime) \\
        &+ \frac{1}{2} w_t(\widehat \pi)\ddot w_t^\prime(\widehat \pi^\prime)(\widehat \pi^\prime - \pi^\prime)^2 + \frac{1}{2} w_t^\prime(\widehat \pi^\prime)\ddot w_t(\widehat \pi)(\widehat \pi - \pi)^2 \\
        &+ o\{ (\widehat \pi - \pi)^2 \} + o\{ (\widehat \pi^\prime + \pi^\prime)^2 \}.
    \end{align*}
    The first three summands on the RHS cancel with the first three summands on the RHS of the original expression, yielding the result.
\end{proof}

\subsubsection*{Proof of Lemma~\ref{lem:eif}}

\begin{proof}
   Lemmas~\ref{lem:eif-helper-1}, \ref{lem:eif-helper-2}, \ref{lem:eif-helper-3}, and \ref{lem:eif} imply that $\bbE \{ \widehat \varphi(Z) - \varphi(Z) \}$ is a second-order product of errors in nuisance functions. By the same argument, the functional satisfies a von Mises expansion with second-order remainder term. The result will then follow by \citet[Lemma 2]{kennedy2023semiparametric}, combined with the fact that $\bbV \{ \varphi(Z) \}$ is bounded. 
   
   \bigskip
   
   To conclude, we must justify that the variance of $\varphi(Z)$ is bounded. Because each summand is conditionally mean zero, the cross-covariances are zero and we can focus on just the variance of each summand. Because the sequential regression functions are almost surely bounded and $\frac{w_t w_t^\prime}{\pi_t \pi_t^\prime}$ is bounded by the construction of the weights, all terms are bounded besides the second term in $\varphi_w$. That second term is
   \[
   \left\{ \prod_{s=1}^{t-1} r_s^\prime(A_s, \overline X_s) \rho_s(\overline X_s) \right\} m_t(\overline X_t) w_t\{ \pi_t(\overline X_t) \} \rho_t(\overline X_t) \phi_t^\prime(A_t, \overline X_t).
   \]
   The variance of this term is bounded by first noting that it can be bounded by \( \bbE \left( \prod_{s=1}^{t} \rho_s^2 \mid \prod_{s=1}^{t} \rho_s < \infty \right) \) by the argument above and the assumed boundedness of $m_t$ and because the weights enforce $\rho_s < \infty$.  By iterated expectations, we have
   \[
   \bbE \left( \prod_{s=1}^{t} \rho_s^2  \mid \prod_{s=1}^{t} \rho_s < \infty \right) = \bbE \left\{ \prod_{s=1}^{t-1} \rho_s^2 \bbE \left( \rho_t^2 \mid \overline X_{t-1}, \rho_t < \infty \right)  \mid \prod_{s=1}^{t-1} \rho_s < \infty \right\}.   
   \]
   Then, by conditional H\"{o}lder's inequality,
   \[
   \bbE \left( \rho_t^2 \mid \overline X_{t-1}, \rho_t < \infty \right) \leq \bbE \left\{ \left( \rho_t^2 \right)^{\tfrac{2+\delta}{2}} \mid \overline X_{t-1}, \rho_t < \infty \right\}^{\tfrac{2}{2+\delta}} = \bbE  \left( \rho_t^{2+\delta} \mid \overline X_{t-1}, \rho_t < \infty \right)^{\tfrac{2}{2+\delta}},
   \]
   which is almost surely bounded by assumption in the lemma. Hence,
   \[
   \bbE \left( \prod_{s=1}^{t} \rho_s^2 \mid \prod_{s=1}^{t} \rho_s < \infty \right) \lesssim \bbE \left( \prod_{s=1}^{t-1} \rho_s^2 \mid \prod_{s=1}^{t-1} \rho_s < \infty \right).   
   \]
   Repeating this argument $t-2$ times, and noting that $\rho_1 = 1$, implies that $\bbV \{ \varphi(Z) \} $ is bounded, which yields the result.
\end{proof}

\subsection*{Proof of Theorem~\ref{thm:convergence}}

\begin{proof}
    We begin with the typical decomposition:
    \begin{align*}
        \widehat \psi(\overline a_T) - \psi(\overline a_T) &= \bbP_n \{ \widehat \varphi(Z) \} - \bbP \{ \varphi(Z) \} \\
        &= (\bbP_n - \bbP) \{ \varphi(Z) \} + (\bbP_n - \bbP) \{ \widehat \varphi(Z) - \varphi(Z) \} + \bbE \left\{ \widehat \varphi(Z) - \varphi(Z) \right\}.
    \end{align*}
    where the first line follows by definition, the second by adding zero and because $\bbP \{ \varphi(Z) \} = 0$. The second term is $o_\bbP(n^{-1/2})$ by Chebyshev's inequality and the assumption that $\lVert \widehat \varphi - \varphi \rVert = o_\bbP(1)$ (cf. \citet[Lemma 2]{kennedy2020sharp}). Meanwhile, the third term is the bias. It satisfies the decomposition in Lemma~\ref{lem:eif-helper-3}; as a reminder, this is:
    \begin{align*}
        \bbE \{ \widehat \varphi(Z) - \varphi(Z) \} &= \sum_{t=1}^{T}  \bbE \left\{ \left( \prod_{s=1}^{t} \widehat r_s - \prod_{s=1}^{t}  r_s \right) \left( \widetilde m_t - \widehat m_t \right) \right\} \\
        &+ \sum_{t=1}^{T} \bbE \Bigg\{ \left( \prod_{s=1}^{t-1} \widehat r_s - \prod_{s=1}^{t-1} r_s \right) \widehat m_t \bbE \left( \widehat \phi_t \mid \overline X_t, \overline a_{t-1} \right) \widehat w_t^\prime \Bigg\} \\
        &+ \sum_{t=1}^T \bbE \Bigg\{ \left( \prod_{s=1}^{t-1} \widehat r_s^\prime \widehat \rho_s - \prod_{s=1}^{t-1} r_s^\prime \rho_s \right) \widehat m_t \widehat w_t \widehat \rho_t \bbE \left( \widehat \phi_t^\prime \mid \overline a_{t-1}^\prime, \overline X_t \right) \Bigg\} \\
        &+ \sum_{t=1}^{T}\bbE \Bigg\{ \left( \prod_{s=1}^{t-1} r_s^\prime \rho_s \right) \widehat m_t \widehat w_t (\widehat \rho_t - \rho_t) \bbE \left( \widehat \phi_t^\prime \mid \overline a_{t-1}^\prime, \overline X_t \right) \Bigg\} \\
        &+ \sum_{t=1}^{T} \bbE \left[ \left( \prod_{s=1}^{t-1} r_s \right) \widehat m_t \left\{ \bbE \left( \widehat \phi_t \mid \overline X_t, \overline a_{t-1} \right) \widehat w_t^\prime + \widehat w_t \bbE \left( \widehat \phi_t^\prime \mid \overline a_{t-1}^\prime, \overline X_t \right) + \widehat w_t \widehat w_t^\prime - w_t w_t^\prime \right\} \right].
    \end{align*}
    The first, second, and fifth lines in the display above can be addressed quite easily: by the boundedness conditions on $\widehat m_t$ and $m_t$, Lemma~\ref{lem:eif-helper-4}, H\"{o}lder's inequality, the triangle inequality, and Cauchy-Schwarz, we can get products of $L_2$ errors:
    \begin{align*}
         &\left| \sum_{t=1}^{T} \bbE \left\{ \left( \prod_{s=1}^{t} \widehat r_s - \prod_{s=1}^{t}  r_s \right) \left( \widetilde m_t - \widehat m_t \right) \right\} \right| \lesssim \sum_{t=1}^T \sum_{s=1}^t \| \widehat r_s - r_s \| \| \widehat m_t - \widetilde m_t \|, \\
         &\left| \sum_{t=1}^{T} \bbE \Bigg\{ \left( \prod_{s=1}^{t-1} \widehat r_s - \prod_{s=1}^{t-1} r_s \right) \widehat m_t \bbE \left( \widehat \phi_t \mid \overline X_t, \overline a_{t-1} \right) \widehat w_t^\prime \Bigg\} \right| \lesssim \sum_{t=1}^T \sum_{s=1}^{t-1} \| \widehat r_s - r_s \| \| \widehat \pi_t - \pi_t \| \text{, and} \\
         &\sum_{t=1}^{T} \bbE \left[ \left( \prod_{s=1}^{t-1} r_s \right) \widehat m_t \left\{ \bbE \left( \widehat \phi_t \mid \overline X_t, \overline a_{t-1} \right) \widehat w_t^\prime + \widehat w_t \bbE \left( \widehat \phi_t^\prime \mid \overline a_{t-1}^\prime, \overline X_t \right) + \widehat w_t \widehat w_t^\prime - w_t w_t^\prime \right\} \right] \lesssim \sum_{t=1}^{T} \| \widehat \pi_t - \pi_t \|^2.
    \end{align*}
    The slightly more complicated terms are the third and fourth terms, which involve $\rho_s$. They are more complicated because we don't have uniform boundedness on $\rho_s$. Instead, we have a bounded $2+\delta$ moment. We first consider the third line from the original display:
    \begin{align*}
        &\left| \bbE \Bigg\{ \left( \prod_{s=1}^{t-1} \widehat r_s^\prime \widehat \rho_s - \prod_{s=1}^{t-1} r_s^\prime \rho_s \right) \widehat m_t \widehat w_t \widehat \rho_t \bbE \left( \widehat \phi_t^\prime \mid \overline a_{t-1}^\prime, \overline X_t \right) \Bigg\} \right| \\
        &\hspace{1in}\leq \sum_{s=1}^{t-1} \left| \bbE \Bigg[ \left( \prod_{u<s} r_u^\prime \rho_u \right) \left\{ (\widehat r_s^\prime - r_s^\prime) \widehat \rho_s + r_s^\prime(\widehat \rho_s - \rho_s) \right\} \left( \prod_{u>s}^{t-1} \widehat r_u^\prime \widehat \rho_u \right) \widehat m_t \widehat w_t \widehat \rho_t \bbE \left( \widehat \phi_t^\prime \mid \overline a_{t-1}^\prime, \overline X_t \right) \Bigg] \right| \\
        &\leq \sum_{s=1}^{t-1} \bbE \left\{ \left| \left( \prod_{u<s} \rho_u \right) \widehat \rho_s \left( \prod_{u>s}^{t-1} \widehat \rho_u \right) \widehat \rho_t (\widehat r_s^\prime - r_s^\prime) (\widehat \pi_t^\prime - \pi_t^\prime) \right| \right\} + \bbE \left\{ \left| \left( \prod_{u<s} \rho_u \right) \left( \prod_{u>s}^{t-1} \widehat \rho_u \right) \widehat \rho_t (\widehat \rho_s - \rho_s)  (\widehat \pi_t^\prime - \pi_t^\prime) \right| \right\} \\
        &\leq \sum_{s=1}^{t-1} \bbE \left\{ \left( \prod_{u<s} \rho_u^2 \right) \widehat \rho_s^2 \left( \prod_{u>s}^{t-1} \widehat \rho_u^2 \right) \widehat \rho_t^2 \right\}^{1/2} \bbE\left\{ (\widehat \pi_t^\prime - \pi_t^\prime)^2 (\widehat r_s^\prime - r_s^\prime)^2 \right\}^{1/2} \\
        &\hspace{1in}+ \bbE \left\{ \left( \prod_{u<s} \rho_u^2 \right) \left( \prod_{u>s}^{t-1} \widehat \rho_u^2 \right) \widehat \rho_t^2 \right\}^{1/2} \bbE \left\{ (\widehat \rho_s - \rho_s)^2 (\widehat \pi_t^\prime - \pi_t^\prime)^2 \right\}^{1/2} \\
        &\lesssim \sum_{s=1}^{t-1} \bbE\left\{ (\widehat \pi_t^\prime - \pi_t^\prime)^2 (\widehat r_s^\prime - r_s^\prime)^2 \right\}^{1/2} + \bbE \left\{ (\widehat \rho_s - \rho_s)^2 (\widehat \pi_t^\prime - \pi_t^\prime)^2 \right\}^{1/2} \\
        &\leq \sum_{s=1}^{t-1} \| \widehat \pi_t^\prime - \pi_t^\prime \|_4 \| \widehat r_s^\prime - r_s^\prime \|_4 + \| \widehat \pi_t^\prime - \pi_t^\prime \|_4 \|\widehat \rho_s - \rho_s \|_4
    \end{align*}
    where the first inequality follows by telescoping and the triangle inequality, the second line by H\"{o}lder's inequality ($\ell_\infty-\ell_1$) and the triangle inequality, the third inequality by Cauchy-Schwarz, the fourth inequality by the same argument as in the proof of Lemma~\ref{lem:eif}, but also using the moment bound on $\widehat \rho_t$ and noting that $r_s^\prime$ and $\widehat r_s^\prime$ and $\widehat w_t$ guaranteed that none of the covariate density ratios (estimated or true) were infinite, and the final inequality follows by Cauchy-Schwarz again, where $\| f(Z) \|_4 = \bbE \{ f(Z)^4 \}^{1/4}$ denotes the $\ell_4(\bbP)$ norm.

    \bigskip

    The same type of argument applies to the fourth line from the original display:
    \begin{align*}
        \left| \bbE \Bigg\{ \left( \prod_{s=1}^{t-1} r_s^\prime \rho_s \right) \widehat m_t \widehat w_t (\widehat \rho_t - \rho_t) \bbE \left( \widehat \phi_t^\prime \mid \overline a_{t-1}^\prime, \overline X_t \right) \Bigg\} \right| &\lesssim \| \widehat \rho_t - \rho_t \|_4 \| \widehat \pi_t^\prime - \pi_t^\prime \|_4. 
    \end{align*}
    Combining all terms yields the bounds on the bias:
    \begin{align*}
        \left| \bbE \left\{ \widehat \varphi(Z) - \varphi(Z) \right\} \right| &\lesssim \sum_{t=1}^T \sum_{s=1}^t \| \widehat r_s - r_s \| \| \widehat m_t - \widetilde m_t \| \\
        &+ \sum_{t=1}^{T} \sum_{s=1}^{t-1} \left\|  \widehat r_s -  r_s \right\| \left\| \widehat \pi_t - \pi_t \right\| \\
        &+ \sum_{t=1}^{T} \| \widehat \pi_t - \pi_t \|^2 \\
        &+ \sum_{t=1}^{T} \sum_{s=1}^{t-1} \| \widehat \pi_t^\prime - \pi_t^\prime \|_4 \| \widehat r_s^\prime - r_s^\prime \|_4 + \| \widehat \pi_t^\prime - \pi_t^\prime \|_4 \|\widehat \rho_s - \rho_s \|_4 \\
        &+ \sum_{t=1}^{T} \| \widehat \rho_t - \rho_t \|_4 \| \widehat \pi_t^\prime - \pi_t^\prime \|_4.
    \end{align*}
    By the conditions of the theorem, $\left| \bbE \left\{ \widehat \varphi(Z) - \varphi(Z) \right\} \right| = o_\bbP(n^{-1/2}).$

    \medskip
    
    Returning to the main decomposition, we then have
    \[
    \sqrt{\frac{n}{\bbV\{ \varphi(Z) \}}}\left\{ \widehat \psi(\overline a_T) - \psi(\overline a_T) \right\} = \sqrt{\frac{n}{\bbV\{ \varphi(Z) \}}} (\bbP_n - \bbP)\{ \varphi(Z) \} + o_\bbP(1) \indist N(0,1)
    \]
    by the central limit theorem.

    \medskip

    Finally, note that $\widehat \sigma^2 \inprob \bbV \{ \varphi(Z) \}$ because $\| \widehat \varphi - \varphi \| = o_\bbP(1)$. Therefore, the result follows by Slutsky's theorem.
\end{proof}

\end{document}